%% file: main.tex
\documentclass[11pt,reqno]{article}

\usepackage[disable,color=green!40,prependcaption,textsize=tiny]{todonotes} 
\newcommand{\topic}[1]{\todo[color=red!40,noline]{#1}}

\usepackage{soul}
\usepackage{etoolbox}
\makeatletter
\patchcmd{\SOUL@ulunderline}{\dimen@}{\SOUL@dimen}{}{}
\patchcmd{\SOUL@ulunderline}{\dimen@}{\SOUL@dimen}{}{}
\patchcmd{\SOUL@ulunderline}{\dimen@}{\SOUL@dimen}{}{}
\newdimen\SOUL@dimen
\makeatother

\sethlcolor{green!0} 

\usepackage{amsthm,amsmath,mathtools,amssymb}
\usepackage[affil-it]{authblk}
\usepackage{csquotes}
\usepackage{enumitem}

\usepackage{diagbox}

\usepackage[colorlinks=true,citecolor=blue,urlcolor=blue,linkcolor=blue]{hyperref} 
\usepackage[nameinlink,capitalize]{cleveref}
\Crefname{table}{Table}{Tables}
\Crefname{ass}{Assumption}{Assumptions}
\usepackage[]{natbib}
\usepackage{autonum} 

\usepackage{array}
\newcolumntype{$}{>{\global\let\currentrowstyle\relax}}
\newcolumntype{^}{>{\currentrowstyle}}
\newcommand{\rowstyle}[1]{\gdef\currentrowstyle{#1}%
  #1\ignorespaces
}

\usepackage{diagbox}
\usepackage{threeparttable}

\usepackage{booktabs,graphicx}
\usepackage{bbm,bm,xspace}
\usepackage[labelfont=bf,textfont=md,font=small]{caption}
\usepackage{subcaption}
\usepackage[margin=1in]{geometry}

\usepackage[doublespacing]{setspace}

\usepackage[indentafter]{titlesec}
\titleformat{name=\section}{}{\thetitle.}{0.8em}{\centering\scshape}
\titleformat{name=\subsection}[runin]{}{\thetitle.}{0.5em}{\bfseries}[.]
\titleformat{name=\subsubsection}[runin]{}{\thetitle.}{0.5em}{\itshape}[.]
\titleformat{name=\paragraph,numberless}[runin]{}{}{0em}{}[.]
\titlespacing{\paragraph}{0em}{0em}{0.5em}
\titleformat{name=\subparagraph,numberless}[runin]{}{}{0em}{}[.]
\titlespacing{\subparagraph}{0em}{0em}{0.5em}

\usepackage{algorithm}
\usepackage{algpseudocode}
\MakeRobust{\Call}

\renewcommand{\P}[2][]{\ensuremath{\mathbb{P}_{#1}\left(#2\right)}\xspace}
\newcommand{\E}[2][]{\ensuremath{\mathbb{E}_{#1}\left(#2\right)}}
\newcommand{\V}[2][]{\ensuremath{\operatorname{Var}_{#1}\left(#2\right)}}
\newcommand{\C}[2][]{\ensuremath{\operatorname{Cov}_{#1}\left(#2\right)}}
\newcommand{\I}[2][]{\ensuremath{\mathbbm{1}_{#1}\left\{#2\right\}}}
\newcommand{\ER}[0]{Erd\H{o}s-R\'{e}nyi\xspace}

\def\one{\ensuremath{\mathbbm{1}}}

\DeclarePairedDelimiter\abs{\lvert}{\rvert}
\DeclarePairedDelimiter\norm{\lVert}{\rVert}
\DeclarePairedDelimiterX\setc[2]{\{}{\}}{\,#1 \;\delimsize\vert\; #2\,}

\makeatletter
\let\oldabs\abs
\def\abs{\@ifstar{\oldabs}{\oldabs*}}
\let\oldnorm\norm
\def\norm{\@ifstar{\oldnorm}{\oldnorm*}}
\let\oldsetc\setc
\def\setc{\@ifstar{\oldsetc}{\oldsetc*}}
\makeatother

\theoremstyle{plain}
\newtheorem{thm}{Theorem}[]
\newtheorem{lem}{Lemma}
\newtheorem{prop}{Proposition}
\newtheorem{cor}{Corollary}

\newtheorem{modl}{Model}

\theoremstyle{definition}
\newtheorem{defn}{Definition}

\newtheorem{ass}{Assumption}

\theoremstyle{remark}
\newtheorem*{rmk}{Remark}

\usepackage{newtxtext}
\usepackage{newtxmath}
\usepackage{relsize}
\let\oldtriangle\triangle
\renewcommand\triangle{\mathlarger{\oldtriangle}}

\begin{document}
\title{Testing for Balance in Social Networks}

\author[1]{Derek Feng}
\author[2]{Randolf Altmeyer}
\author[3]{Derek Stafford}
\author[,3]{Nicholas A. Christakis\footnote{Co-Corresponding Authors}}
\newcommand\CoAuthorMark{\footnotemark[\arabic{footnote}]} 
\author[,1]{Harrison H. Zhou\protect\CoAuthorMark}

\affil[1]{\small Department of Statistics and Data Science\\
Yale University, New Haven, CT, 06520}
\affil[2]{\small Department of Pure Mathematics and Mathematical Statistics \\
University of Cambridge, Cambridge, UK CB3 0WB}
\affil[3]{\small Yale Institute for Network Science\\
Yale University, New Haven, CT 06520}


\maketitle

\begin{abstract}
  Friendship and antipathy exist in concert with one another in real social networks.
  Despite the role they play in social interactions, antagonistic ties are poorly understood and infrequently measured.
  One important theory of negative ties that has received relatively little empirical evaluation is balance theory,
  the codification of the adage ``the enemy of my enemy is my friend'' and similar sayings.
  Unbalanced triangles are those with an odd number of negative ties, and the theory posits that such triangles are rare.
  To test for balance, previous works have utilized a permutation test on the edge signs. The flaw in this method, however, is that it assumes that negative and positive edges are interchangeable.
  In reality, they could not be more different.
  Here, we propose a novel test of balance that accounts for this discrepancy and show that our test is more accurate at detecting balance.
  Along the way, we prove asymptotic normality of the test statistic under our null model, which is of independent interest.
  Our case study is a novel dataset of signed networks we collected from 32 isolated, rural villages in Honduras.
  Contrary to previous results, we find that there is only marginal evidence for balance in social tie formation in this setting.
\end{abstract}

\makeatletter{\renewcommand*{\@makefnmark}{}
\footnotetext{%

  Derek Feng is a Lecturer, Department of Statistics and Data Science, Yale University, New Haven, CT 06520 USA (email: \emph{derek.feng@yale.edu});
  Randolf Altmeyer is a Post-Doc, Department of Pure Mathematics and Mathematical Statistics, University of Cambridge, UK CB3 0WB, (email: \emph{ra591@maths.cam.ac.uk });
  Derek Stafford is a Post-Doc, Yale Institute for Network Science, Yale University, New Haven, CT 06520 USA (email: \emph{derek.stafford@gmail.com});
  Nicholas Christakis is the Sterling Professor of Social and Natural Science, Yale Institute for Network Science (also Department of Sociology, and Department of Medicine), Yale University, New Haven, CT 06520 USA (email: \emph{nicholas.christakis@yale.edu});
  Harrison H. Zhou is the Henry Ford II Professor of Statistics and Data Science, Yale University, New Haven, CT 06520 USA (email: \emph{huibin.zhou@yale.edu}).
  Support for this research was provided by grants from the Bill and Melinda Gates Foundation, the Robert Wood Johnson Foundation, as well as NSF Grant DMS-1507511, and DFG Research Training group 1845 `Stochastic Analysis'.
  The authors would like to thank the two anonymous referees for their insightful comments and feedback, which greatly improved the manuscript.
}\makeatother}



\vfill 


\input{Chapters/introduction}

\input{Chapters/method}

\input{Chapters/theory}

\input{Chapters/simulation}

\input{Chapters/honduras}

\input{Chapters/conclusion}

\appendix
\input{Chapters/appendix}

\bibliographystyle{apalike}
\bibliography{bib,others}

\end{document}

%% file: Chapters/introduction.tex

\section{Introduction}

Models of social network structure generally build on assumptions about myopic agents, whereby global network features emerge from the dynamic local decision rules of individual agents \citep{holland1998emergence,Kossinets:2006je}. For instance, if agents tend to attach to more central or popular actors, scaling emerges in the degree distribution of the graph \citep{Barabasi:1999je};
if people generally form connections with those who are similar, social networks exhibit homophily \citep{mcpherson2001birds};
if agents form infrequent but random connections with other agents, the social graph has a small diameter, following the small-world phenomenon \citep{Watts:1998db}.

All of these models, however, are restrictive in that they only apply to positive ties.
Much less is theorized or known about the fundamental properties of negative ties.
In principle, they need not share the same structural properties as their positive counterparts.
Moreover, as most social graphs are \emph{signed} (i.e. have both positive and negative ties), this raises the question of how the presence of the negative ties affects the surrounding positive network structure, and how we should model them concurrently.


\topic{defining balance}
One important theory of negative ties advanced by \citet{Heider:1946vz} relates to an agent's desire for balance in social relationships \citep{Harary:1959ct,simmel2010conflict}. Balance theory postulates that a need for cognitive consistency leads agents to seek to balance the valence in their local social systems.
Simply stated, friends should have the same friends and the same enemies.
This translates, in graph-theoretic terms, to requiring the product of the signs on a triangle to be positive.
Triangles that violate this property are deemed \emph{unbalanced}, and the theory posits that such triangles should be rare compared to their \emph{balanced} counterparts.

\topic{is balance real?}
Balance theory is very simple to state and almost self-evident in nature.
After all, it has already been assimilated into the wider culture through such aphorisms as \emph{``the enemy of my enemy is my friend''}.
However, as evidenced by the success of behavioral economics \citep{kahneman1979prospect}, human actors
will often act irrationally, even so far as to violate transitivity \citep{tversky1981framing}. It is therefore not unreasonable to envisage people violating transitivity in their social graph.

\topic{little evidence for balance}
As it stands, balance theory has received sparing empirical evaluation.
Tests of balance theory require the observation of antagonistic connections between actors, but these ties are often either ignored when the data is gathered, or simply unavailable due to the unwillingness of the actors themselves to divulge such information.
Those studies which have been able to observe antagonistic ties have often done so in artificial settings -- and have been very liberal about what constitutes an antagonistic tie -- like nominations to adminship on Wikipedia \citep{Leskovec:2010hw}, and user ratings of trustworthiness in an e-commerce website \citep{Guha:2004hs}, rather than in face-to-face settings, with some exceptions \citep{mouttapa2004social,huitsing2012bullying,xia2009exploring}.

\topic{test for balance, why sucks}
Though the underlying datasets may be vastly different, these studies all resort to exactly the same statistical test to verify balance in their signed networks:
for the test statistic, they use the number of balanced triangles as a measure of the degree of balance in a graph;
the null model corresponds to a permutation test on the edge weights of the observed graph.
Drawing samples from the null distribution then reduces to shuffling the signs on the graph.
The simplicity of this null model belies its principal flaw though -- namely, that it treats negative and positive ties as interchangeable.
The problem is that, as we shall soon demonstrate, negative ties behave remarkably like random ties drawn from an \ER graph.
On the other hand, researchers have spent the past few decades documenting the various ways that a network of positive ties differs from an Erd\H{o}s-R\'{e}nyi graph.

Features like preferential attachment and clustering are fundamental to our understanding of positive ties -- features that are clearly absent in negative ties.
Thus, by treating positive and negative ties as exchangeable, this null hypothesis creates a test, not for balance, but for differences in the behavior of positive and negative ties.

\topic{data example}
As an example, consider one of the social networks from our Honduras dataset, shown in \Cref{fig:intro_network}.
Decomposing the graph into its signed subgraphs reveals \subref{subfig:p22} a typical positive social network, and \subref{subfig:n22} a negative subgraph that could be easily mistaken for a sample from a random graph model.
The contrast between the two subgraphs could not be more extreme, and clearly indicates that these two types of ties should not be treated as exchangeable.

\begin{figure}[h]
  \small
  \begin{subfigure}[b]{.5\textwidth}
    \centering
    \includegraphics[width=1.2\textwidth]{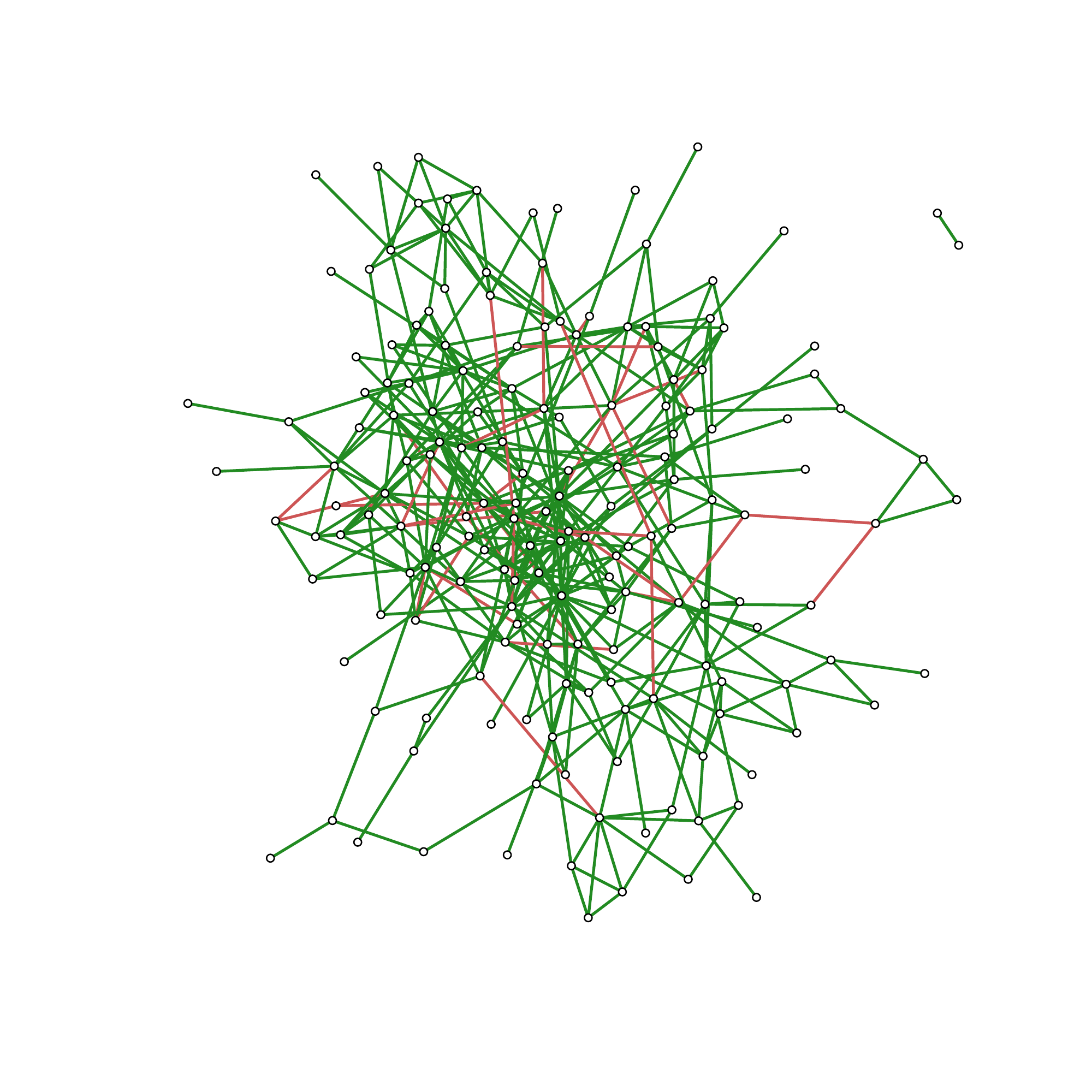}
    \vspace*{2em}
    \subcaption{Signed Graph}
    \label{subfig:u22}
  \end{subfigure}
  \begin{subfigure}[b]{.5\textwidth}
    \centering
    \includegraphics[width=.9\textwidth]{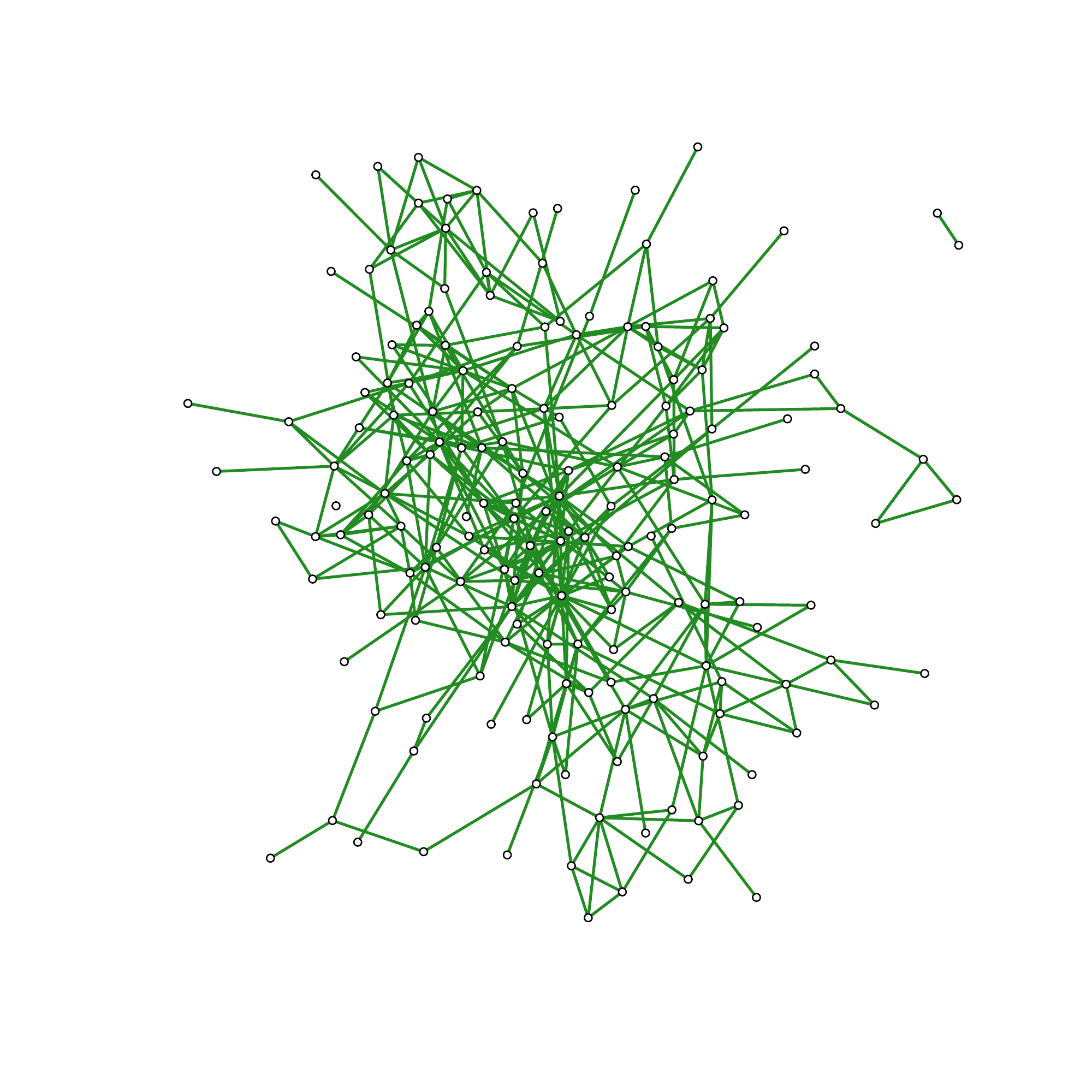}
    \vspace*{-2em}
    \subcaption{Positive Subgraph}
    \label{subfig:p22}
    \centering
    \includegraphics[width=.8\textwidth]{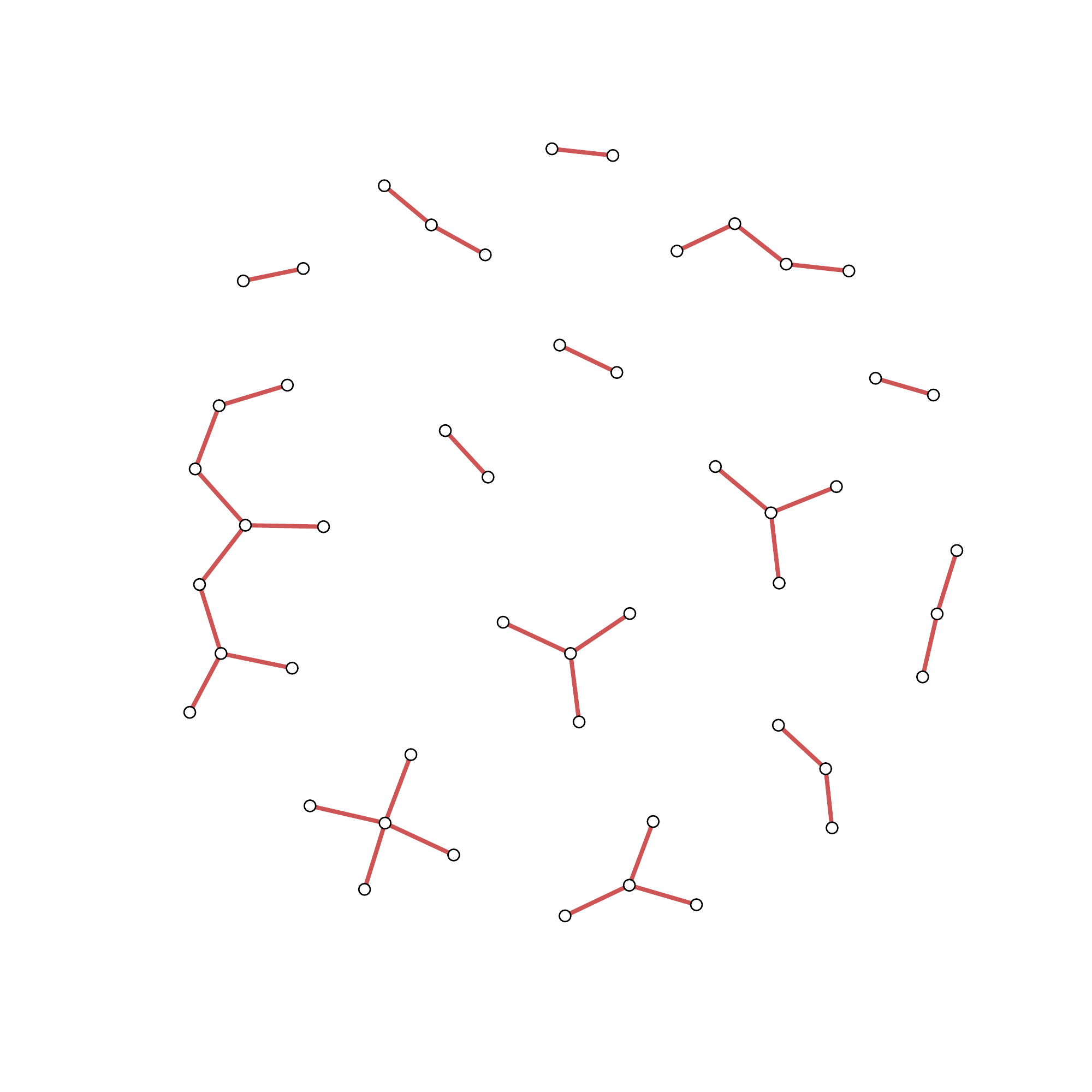}
    \vspace*{-1em}
    \subcaption{Negative Subgraph}
    \label{subfig:n22}
  \end{subfigure}
  \caption{\subref{subfig:u22} A signed social network (Village \# 22).
  Much has been studied about the positive subgraph \subref{subfig:p22}, but very little is known about the negative subgraph \subref{subfig:n22}. The example here suggests that the negative subgraph behaves more like a random graph.}
  \label{fig:intro_network}
\end{figure}

\topic{careful of null hypothesis}
As the ``replication crisis'' \citep{open2015estimating} plays out in the scientific community, and notions like p-hacking and ``garden of forking paths'' \citep{gelman2013garden} become mainstream, one overlooked but equally important issue
is the selection of an appropriate null hypothesis.
In most settings, there is a canonical choice for the null hypothesis.
On the other hand, in the context of complex social data such as social networks, we no longer have this luxury,
as it is difficult to accurately model such data with a simple parametric model.
In this regime, it is imperative to choose null models carefully.


\topic{our contribution: new null model}
The main contribution of this paper is to provide a new null model that resolves the issues raised above.
The crux of the solution is the following key observation: a crucial way in which negative and positive ties differ is through their \emph{embeddedness} level (the number of triangles that tie is a member of) --
transitivity and homophily encourage higher levels of embeddedness in positive ties.
For instance, the average embeddedness of positive ties and negative ties in the network shown in \Cref{fig:intro_network} was $1.2$ and $0.5$, respectively.
Our new method, therefore, is to stratify the permutation across embeddedness levels, thereby ensuring that the embeddedness profiles of negative and positive ties remain invariant. This preserves the fundamental differences between the two kinds of ties, creating a more accurate null model of a signed social network without balance.
This is supported by both our simulation studies and our theoretical results, where we show that for a reasonable definition of absence of balance in a graph,
the true type-I error rate of the old test converges to 1 while the type-I error rate for the new test is consistent with the specified $\alpha$.


\topic{describe theory}
To compare the relative performance of the two tests, we show asymptotic normality of the test statistic under the two null models.
Due to the stratified nature of the permutation, this is a nontrivial result, and, to the best of our knowledge, this is the first result showing asymptotic normality of this type of graph statistic under a stratified permutation model.
The key insight is that a distribution derived from a permutation test -- even a stratified permutation -- can be obtained as conditional distribution of independent random variables.
This is similar to the dichotomy between the $G(n,p)$ and $G(n,m)$ random graph model \citep{janson2011random}.
Under certain conditions, the limit and the conditioning operation may be interchanged, enabling us to carry the central limit theorem result in the independent case to the permutation case.
This proof technique of \citet{Janson:2007dx} has wide applicability, not least in the nascent field of (nonparametric) inference on random graphs.

Our final contribution is that we analyze a comprehensive dataset capturing both positive and negative ties between individuals in a social network -- namely, the networks of 32 villages in rural Honduras \citep{kim2015social}.
This novel dataset provides a first look into the behavior of interaction between negative and positive human relationships.
We find that negative ties behave very differently from positive ties.
Applying our new test of balance to the village networks reveals that balance barely registers as an underlying mechanism dictating the structure of signed networks, which is contrary to the conclusions drawn from the previous literature.

\subsection{Organization}

The rest of the paper is organized as follows.
In \Cref{sec:method}, we formally introduce the notion of balance, describe the old null model, its fundamental flaws, and our proposed new null model.
We prove asymptotic normality of the test statistic under both null models in \Cref{sec:theory}, and from this, we show that the new test has a lower type-I error rate than the old test.
This is also supported by the simulation studies we perform in \Cref{sec:simulations}.
Finally, in \Cref{sec:data}, we analyze the networks of 32 rural villages in Honduras.

%% file: Chapters/method.tex

\section{Method}
\label{sec:method}



\topic{defining balance and the consequences}
The theory of balance is an old theory, predating many of the ``classic'' celebrated ideas in social network analysis.
First proposed in \citet{Heider:1946vz}, it was later made formal by \citet{Cartwright:1956dh},
who recast the theory into the more natural graph-theoretic framework.
Adopting such a framework, let us first fix some notation.
Assume that we are given an undirected signed graph $G = (V,E,W)$, where
\begin{itemize}
  \item $V,E$ are vertex and edge sets, respectively, with sizes given by $|V| = N$, $|E| = n$;
  \item $W \in \left\{ -1, +1 \right\}^{n}$ is a vector of edge signs.
\end{itemize}
Let $\triangle, \triangle'$ be the ordered and unordered triplet of indices that form a triangle in $G$, respectively:
\begin{align}
  \triangle &\coloneqq \left\{ (i,j,k) \in E^{3}: i,j,k \text{ form a triangle} \right\}, \\
  \triangle' &\coloneqq  \left\{ (i,j,k) \in E^{3}: i,j,k \text{ form a triangle};\, i < j < k \right\}.
\end{align}
Then, the number of triangles in $G$ is equal to $\abs{\triangle'}$, while $\abs{\triangle} = 6 \abs{\triangle'}$.
For a triangle $(i,j,k)$ in $\triangle$, we say it is \textit{balanced} if $W_i W_j W_k = 1$, and unbalanced otherwise.
Specifically, the unbalanced triangles are those with one or three negative ties ($t_1$ and $t_3$ in \Cref{fig:baltriangles}, respectively), while the remaining two are balanced.
The graph $G$ is then deemed \textit{balanced} if all the realized triangles of $G$ are balanced -- that is,
\begin{align}
  G \text{ is balanced} \iff W_i W_j W_k = 1, \quad \forall (i,j,k) \in \triangle
\end{align}
A simple consequence of this definition is that $G$ is balanced if and only if it can be decomposed into two positive subgraphs that are joined by only negative edges (see \citet{Cartwright:1956dh}).
Finally, a key property of edges that will play a central role in our analysis is its embeddedness, which we define below.
\begin{defn}
  The \emph{embeddedness} of an edge $i \in E$, which we denote by $\varepsilon_{i}$, is the number of triangles that edge $i$ is a part of:
  \begin{align}
    \varepsilon_{i} \coloneqq \frac{1}{2} \sum_{(j,k) \in E^2} \I{ (i,j,k) \in \triangle}.
  \end{align}
\end{defn}

\begin{figure}[h]
  \small
  \begin{minipage}{.25\textwidth}
    \centering
    \includegraphics[width=0.9\linewidth]{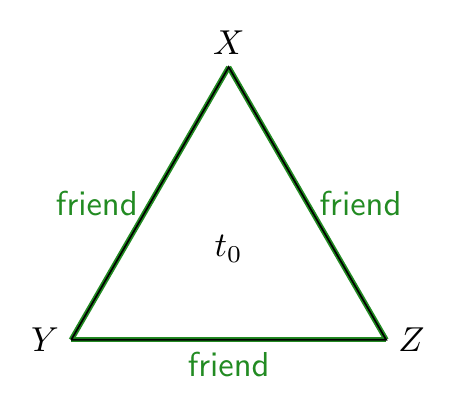}\\
  \end{minipage}%
  \begin{minipage}{.25\textwidth}
    \centering
    \includegraphics[width=0.9\linewidth]{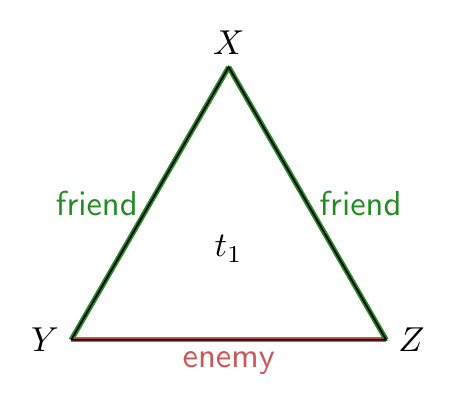}\\
  \end{minipage}
  \begin{minipage}{.25\textwidth}
    \centering
    \includegraphics[width=0.9\linewidth]{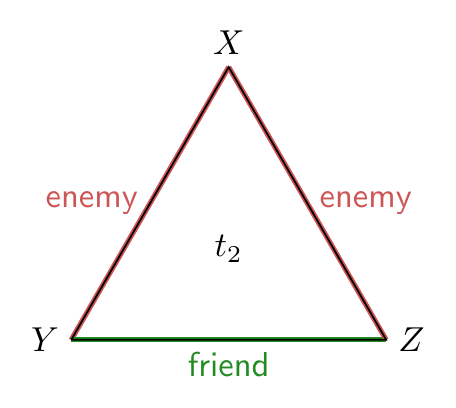}\\
  \end{minipage}%
  \begin{minipage}{.25\textwidth}
    \centering
    \includegraphics[width=0.9\linewidth]{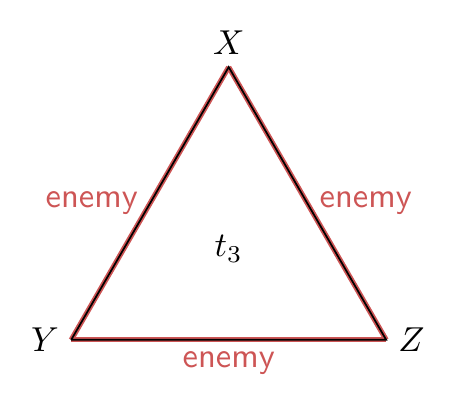}\\
  \end{minipage}
  \caption{The four possible arrangements of positive and negative ties in a triangle with undirected ties. Balance theory states that $t_1$ and $t_3$, with an uneven number of negative ties, are unbalanced.}
  \label{fig:baltriangles}
\end{figure}

\topic{psychological background for balance}
One theorized psychological mechanism for balance theory is \emph{cognitive dissonance}, the state of mental strain a person experiences while holding two conflicting beliefs.
This theory holds that it requires high cognitive load to feel both animosity and goodwill towards others \citep{festinger1962theory}.
For instance, in the case of $t_1$ (\Cref{fig:baltriangles}), $Y$ holds negative feelings towards $Z$, but also, by transitivity through $X$, $Y$ should possess positive feelings towards $Z$, leading to cognitive dissonance.
Balance theory posits that, to avoid the cognitive strain from imbalance, individuals ($Y$) will take measures to resolve such inconsistencies in their local social system,
such as switching the valence of their own relationships (befriend $Z$ instead).
On the other hand, agents may choose to simply ignore the facts that are in conflict.
The literature on \emph{cognitive biases} --
e.g. the work of \citet{tversky1981framing} who showed that people are capable of violating the transitivity of their own preferences
-- suggests that individuals might be unperturbed by, or simply unaware of, such inconsistencies or lack the willpower to correct their local systems.

Here, we adopt the frequentist approach favored by the existing literature and devise a hypothesis test for balance theory.
Two ingredients are needed to specify such a test:
a test statistic that measures the balance on a graph, and a null model that describes a graph without balance.
We begin with the choice of a measure of balance.
It is immediately obvious that the binary definition of balance from \citet{Cartwright:1956dh} is highly impractical, as almost every social graph is equally ``not balanced'' under this definition, even though some graphs are clearly more balanced than others.
A more appropriate measure, and the one used throughout this literature, is the number of unbalanced triangles, as higher levels of balance should result in fewer unbalanced triangles.\footnote{The ratio of unbalanced triangles to all triangles would seem like a more appropriate statistic, but since we are ultimately performing a statistical test with a null model that leaves the number of triangles invariant, using the ratio is therefore equivalent to using just the numerator.}
This count, however, is meaningless on its own.

The second and more important choice is the null model that describes a graph without balance, and it is here that we diverge from the current literature. 
Despite the broad spectrum of application settings, the existing literature is unanimous in its choice of a null model.
This null model corresponds to one derived from a \emph{permutation test}, where the permutation is over the signs of the edges on the graph.
To generate a graph under this null model, we simply shuffle the signs on the observed graph.
An equivalent formulation is the following generative model. Start with a positive social graph ($G$), and suppose that negative ties are only ever formed by switching the signs of positive ties. With a fixed budget of $m$ negative ties, generate a graph by picking $m$ ties at random to switch.

To see why this describes a graph without balance, note that balance is a statement about the relation between negative and positive ties, and makes claims about their relative positions on the graph.
Thus, a generative model where the sign of an edge is independent of its position on the graph is necessarily absent of balance.

Having chosen a null model (and a test statistic), the test for balance is now fully specified. For an input graph $G$, we test for balance by comparing the number of unbalanced triangles in $G$ against the distribution of the same statistic under the null model.
The details of the general testing procedure are shown in \Cref{alg:balance}, and the details for this test can be found in \Cref{alg:perm_old}.
Note that this is a \emph{left-tailed} test, as the presence of balance should reduce the number of unbalanced triangles.

Formally, the test statistic that we use to measure balance is the number of unbalanced triangles, given by
\begin{align}
  U \coloneqq \sum_{(i,j,k) \in \triangle'} \I{W_{i} W_{j} W_{k} = -1}.
\end{align}
Denote by $\tau$ a random variable which is uniformly distributed over the set of all permutations of $E$ -- namely over the symmetric group $S_n$ (of size $n!$). This enables us to define a random weight vector $W_{\tau}$ by $(W_{\tau})_i = W_{\tau(i)}$.
Then, the old null model is equivalent to the graph-valued random variable given by $G_{\tau} = (V, E, W_{\tau})$
\footnote{This uniform random permutation is technically \hl{overkill}, as it actually only has an effective size of $\binom{n}{m}$, where $m$ is the number of negative ties, since we are only permutating the signs around.}.
We shall use $G_{\tau}$ interchangeably to mean either the random variable or the associated probability distribution over signed graphs.

\begin{modl}[Old Model]\label{modl:old}
  The graph has distribution $G_{\tau}$, where $\tau$ is uniformly distributed over $S_{n}$.
\end{modl}
The old null hypothesis then corresponds to $H_{0}^{\tau}\colon \mathbf{G} \sim G_{\tau}$, while the number of unbalanced triangles now takes the form
\begin{align}
  U_{\tau} \coloneqq \sum_{(i,j,k) \in \triangle'} \I{ W_{\tau(i)} W_{\tau(j)} W_{\tau(k)} = -1 }.
\end{align}

\begin{algorithm}
\small
\caption{General Test for Balance}\label{alg:balance}
\begin{algorithmic}[1]
\Procedure{TestBalance}{$G$, $N$, \textsc{NullModel}} \Comment{$N$ is the number of simulations}
  \State $r \gets$ \Call{TestStatistic}{$G$} \Comment{observed statistic}
  \For{$i\gets 1$ to $N$}
    \State $\hat{G} \gets$ \Call{NullModel}{$G$} \Comment{generate a graph under the null model}
    \State $s[i] \gets $ \Call{TestStatistic}{$\hat{G}$} \Comment{calculate the statistic on the new graph}
  \EndFor
  \State \Return the fraction of $s \leq r$ \Comment{calculate the $p$-value}
\EndProcedure
\Statex
\Function{TestStatistic}{$G$} \Comment{unbalanced triangle count}
  \State \Return the number of triangles in $G$ that have an odd number of negative ties
\EndFunction
\end{algorithmic}
\end{algorithm}

\begin{algorithm}
\small
\caption{The Old Test}\label{alg:perm_old}
\begin{algorithmic}[1]
\Procedure{OldTest}{$G$, $N$}
  \State \Return \Call{TestBalance}{$G$, $N$, \Call{UniformPermute}{$G$}}
\EndProcedure
\Statex
\Procedure{UniformPermute}{$G$}
  \State \Return \Call{PermuteSign}{$G, \operatorname{edges}(G)$}
\EndProcedure
\Statex
\Function{PermuteSign}{$G, E$} \Comment{$E$ is a subset of the edges of $G$}
  \State Permute the signs on the edge set $E$ of $G$ at random
  \State \Return $G$
\EndFunction
\end{algorithmic}
\end{algorithm}

This choice of null model has several flaws, however.
First, it confines the formation of negative ties to switches from the existing positive edges, and so prohibits the formation of negative ties from locations where there are currently no edges.
However, this type of behavior might have \hl{some semblance to reality},
as it can be said that having negative feelings towards someone presupposes that you know them well enough to dislike them. Additionally, if we do not restrict the negative edges to the support, we then need to make a judgement as to how the presence of edges should be modeled, which introduces even more subjectivity to the null model.

A more concerning issue is that this null model treats negative and positive ties as interchangeable.
This raises two problems.
The first is that such an assumption is strictly stronger than assuming there is no balance.
Indeed, the two types of ties being exchangeable is equivalent to the sign of an edge being independent of its location, which necessarily implies that there is no balance.
Hence, this null model is in fact testing a stronger statement than lack of balance, leading to a potentially inflated type-I error.
\todo{try to incorporate the idea of conservative test, extent of conservatism}

The second problem is that such an assumption is unrealistic.
On the one hand, researchers have spent the past few decades cataloging the various ways in which positive social networks differ from \ER graphs.
Social networks from a wide variety of settings have been found to share structural similarities \citep{Apicella:2012dt}, such as degree assortativity, transitivity, and homophily -- all properties patently absent in \ER graphs.
On the other hand, as we demonstrate in \Cref{sec:data}, the negative tie subgraph behaves remarkably similar to a random graph.
Thus, by assuming that negative and positive ties are exchangeable, the existing literature has chosen a test that is essentially condemned to significance.


Our main methodological contribution is that we devise a new null model that addresses the aforementioned problems:
instead of having a uniform permutation across all edges, we stratify the permutation across edges of the same embeddedness.
Formally, the new null hypothesis differs from the old null hypothesis in the choice of random permutation.
Instead of the random permutation $\tau$, which is uniformly distributed over all permutations of $E$, the new random variable, which we denote by $\pi$, is a (disjoint) composition of uniform permutations, one for each level of embeddedness.
\begin{modl}[New Model]\label{modl:new}
  The graph has distribution $G_{\pi}$, where $\pi$ is given by
\begin{align}
  \pi (E) \coloneqq \left( \tau_{L} \circ \ldots \circ \tau_{1} \right) (E),
\end{align}
where $L$ is the maximum embeddedness level of $G$, and $\tau_{l}$ is uniformly distributed over permutations of the set of edges with embeddedness $l$, denoted by $E_l$, and all other edges are untouched.
\end{modl}
\begin{rmk}
  Note that we do not include the permutation $\tau_0$ (targeting edges which are not part of any triangle) in $\pi$, since such a permutation would not change the graph (and in particular, would not change the number of unbalanced triangles).
\end{rmk}
The new null hypothesis is now $H_0^{\pi}\colon \mathbf{G} \sim G_{\pi}$,
and the corresponding test statistic is given by
\begin{align}
  U_{\pi} = \sum_{(i,j,k) \in \triangle'} \I{ W_{\tau_{\varepsilon_i}(i)} W_{\tau_{\varepsilon_j}(j)} W_{\tau_{\varepsilon_k}(k)} = -1 }.
\end{align}
Define $n_l$ to be the number of edges of embeddedness $l$ and $m_l$ to be the number of negative edges of embeddedness $l$:
\begin{align}
  n_l &= \sum_{i \in E} \I{\varepsilon_i = l}, \\
  m_l &= \sum_{i \in E} \I{\varepsilon_i = l, W_i = -1}.
\end{align}
Note that $\pi$ is determined completely by $\left\{ n_l \right\}_{l = 0}^{L}$.

The exact steps of this procedure are described by the function \textsc{StratifiedPermute} in \Cref{alg:perm_new}.
Crucially, this stratified permutation leaves the embeddedness profile of both negative and positive ties invariant.
As a result, we maintain the embeddedness profiles of both types of tie, while still ensuring there is enough flexibility to garner meaningful variability as a model of a signed network without balance.
Essentially, we argue that ties are only exchangeable at the same embeddedness level.

\begin{algorithm}
\small
\caption{The New Test}\label{alg:perm_new}
\begin{algorithmic}[1]
\Procedure{NewTest}{$G$, $N$}
  \State \Return \Call{TestBalance}{$G$, $N$, \Call{StratifiedPermute}{$G$}}
\EndProcedure
\Statex
\Procedure{StratifiedPermute}{$G$}
  \State $tc \gets$ the triangle membership count for each edge \Comment{e.g. $tc = (1, 1, 3, 3, 0, 4)$}
  \State $utc \gets$ $\operatorname{unique}(tc)$, the unique counts in $tc$ \Comment{e.g. $utc = (1,3,0,4)$}
  \For{$t$ in $utc$}
    \State $E \gets$ the edges in $G$ with triangle count $t$
    \State $G \gets$ \Call{PermuteSign}{$G, E$}
  \EndFor
  \State \Return $G$
\EndProcedure
\Statex
\Function{PermuteSign}{$G, E$} \Comment{$E$ is a subset of the edges of $G$}
  \State Permute the signs on the edge set $E$ of $G$ at random
  \State \Return $G$
\EndFunction
\end{algorithmic}
\end{algorithm}

\vskip1em
The computational cost of the two procedures can be broken down into two tasks: the enumeration of all the triangles in the graph, and the determination of the null distribution of the statistic.
The first part is essentially unchanged by the new stratified test.
In particular, it turns out there is a very straightforward way of calculating both the embeddedness level of edge $(i,j)$ and unbalanced triangle counts: letting $A$ be the signed adjacency matrix corresponding to $G$, we have
\begin{align}
  \mathcal{E}_{ij}&=\left(\left|A\right|\circ\left|A\right|^{2}\right)_{ij}, \\
  U_{n}&=\frac{1}{12}\text{tr}\left(\left|A\right|^{3}-A^{3}\right)=\frac{1}{12}\text{tr}\left(\mathcal{E}\one\one'-A^{3}\right),
\end{align}
where $\one \in \mathbb{R}^{N}$ is a vector of all ones.
It is clear, however, that our method will be more computationally intensive compared to the original test for the latter part, as now we must perform (at most) $L$ permutations for each draw from the null distribution.
For large enough graphs, where Monte Carlo simulations are prohibitively expensive, these considerations are modulated by the fact that our limit theorem demonstrates that a normal approximation of the null distribution suffices.

\topic{shortcomings}
A potential shortcoming of our method is that, since it operates at the embeddedness level, networks where many of the embeddedness levels have only one sign edge type will not have an expressive null distribution.
A simple adjustment is the following: instead of permuting edges across embeddedness levels, we can permute across ranges of levels, which allows signs to hop to different levels.
The question then arises of how to construct the bins.

\topic{Example showcasing the difference between the two tests}
To clearly delineate the differences between our new model and the old model, consider a toy example shown in \Cref{fig:toy2}.
In this social network, there is a core group of six individuals that have formed a clique, as well as an additional three individuals on the periphery that have formed unbalanced triangles with the core group.
We claim that this network is not balanced, and so a statistical test should not reject the null that there is no balance. One might argue on the contrary, as there are more balanced triangles than unbalanced ones. However, the clique should not be treated as evidence for balance, as it can already be explained by transitivity.
\todo{come up with a better example? it's hard}
For instance, a graph with no negative ties is ``balanced'' in that there are no unbalanced triangles, but this is clearly not strong evidence supporting the full spectrum of balance. That is why the key triangles to monitor are those with negative ties ($t_1, t_2, t_3$ in \Cref{fig:baltriangles}).

Applying the old test (\cref{subfig:toy2-1}), which ignores embeddedness levels, we see that it allows the negative ties free rein over the entire support of the graph, including to the dense cluster in the middle, where the embeddedness level is 5. This results in an artificially inflated mean of the null distribution.
On the other hand, our test (\cref{subfig:toy2-2}) restricts this edge to those positions with also embeddedness level of 1, shown in blue.
The old test incorrectly rejects the null, while the new test does not.

\begin{figure}[h]
  \small
  \begin{subfigure}[b]{.49\textwidth}
    \hspace*{2em}
    \includegraphics[width=.75\textwidth]{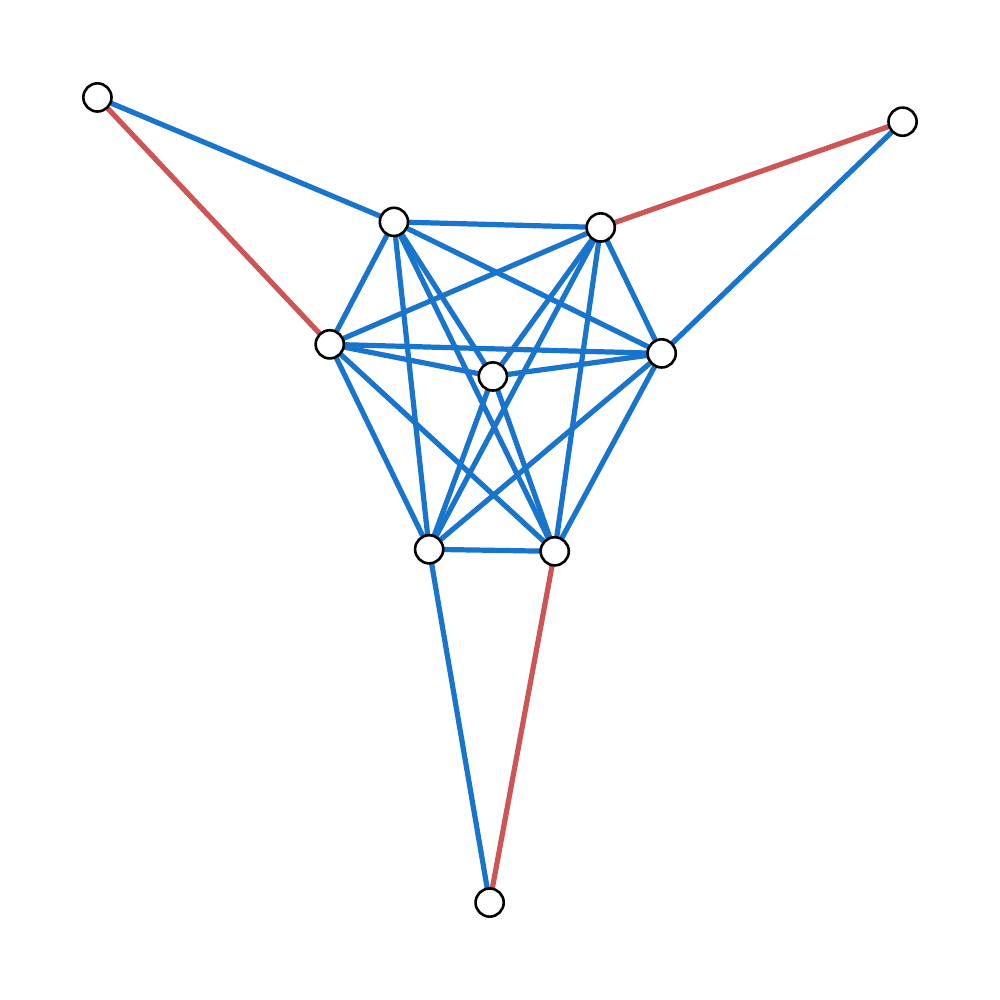}
    \vspace*{-1em}
    \subcaption{Old (Uniform)}
    \label{subfig:toy2-1}
  \end{subfigure}
  \begin{subfigure}[b]{.49\textwidth}
    \hspace*{2em}
    \includegraphics[width=.75\textwidth]{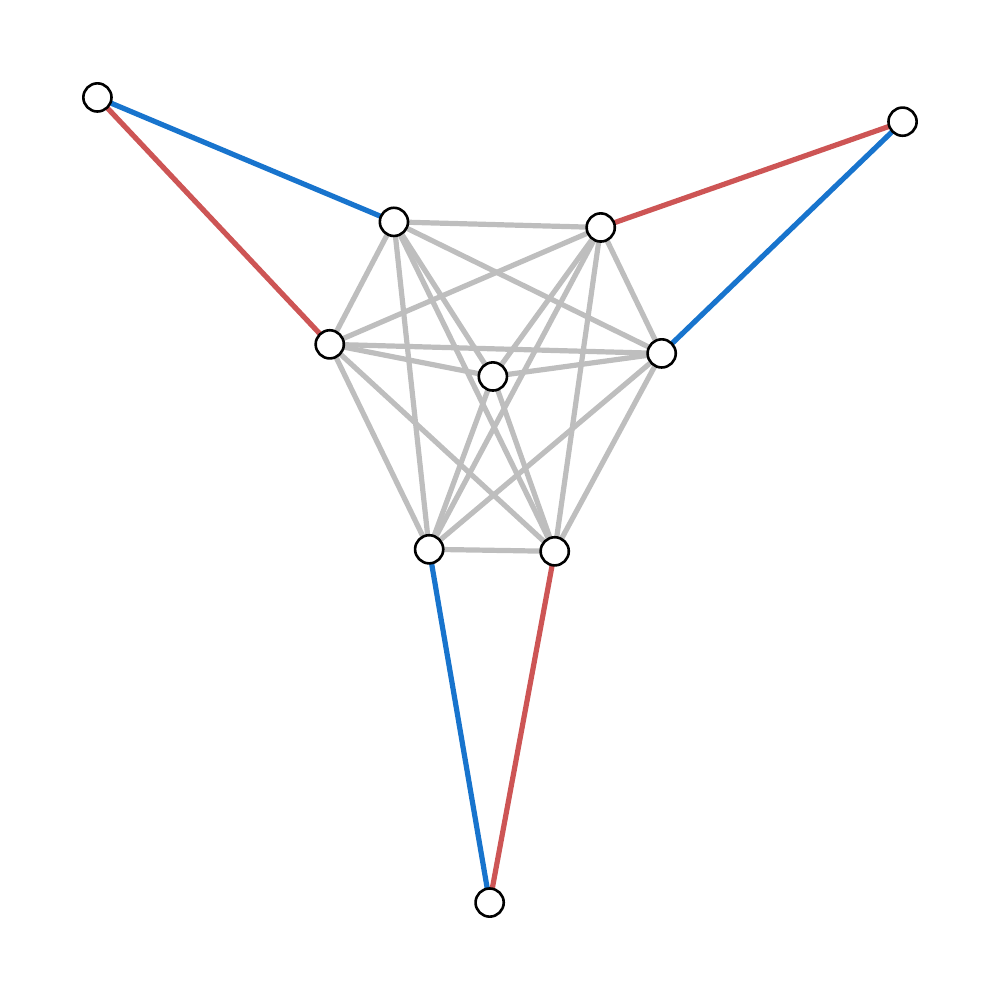}
    \vspace*{-1em}
    \subcaption{New (Stratified)}
    \label{subfig:toy2-2}
  \end{subfigure}
  \caption{Comparison of the two permutations. The blue edges correspond to the positions where the negative (red) edges are able to be moved to under the different tests.}
  \label{fig:toy2}
\end{figure}

\topic{discussion of dynamic vs static}
Finally, we discuss the apparent dynamic nature of balance.
The intuitive interpretation of balance presents itself as a dynamic process, one where individuals notice an unbalanced state and then correct it.
An appropriate test of balance would then involve monitoring the evolution of a graph to see if such local corrections are observed. Ignoring the difficulty of obtaining dynamic graph data, we claim that this dynamic analysis is problematic, as balance is not necessarily expressed in discrete steps.
Consider the first scenario in \Cref{fig:dynamic}, which shows a graph entering and leaving an unbalanced state. If we were to capture all three states, we would flag this as an example of balance at play.
Suppose, however, that the window of time between events 1 and 2 is smaller than the resolution of the graph evolution samples. We would then recover scenario two of \Cref{fig:dynamic} instead, which would not be flagged as evidence for balance under a dynamic model, even though there was a local correction.
Thus, as tempting as it may be to treat balance in terms of dynamic local corrections, we think it is more appropriate to consider balance as a holistic property of a graph.
\hl{Of course, the relative timing of data collection compared to the underlying social process is crucial here.}

\begin{figure}[h]
  \centering
  \includegraphics[width=0.7\linewidth]{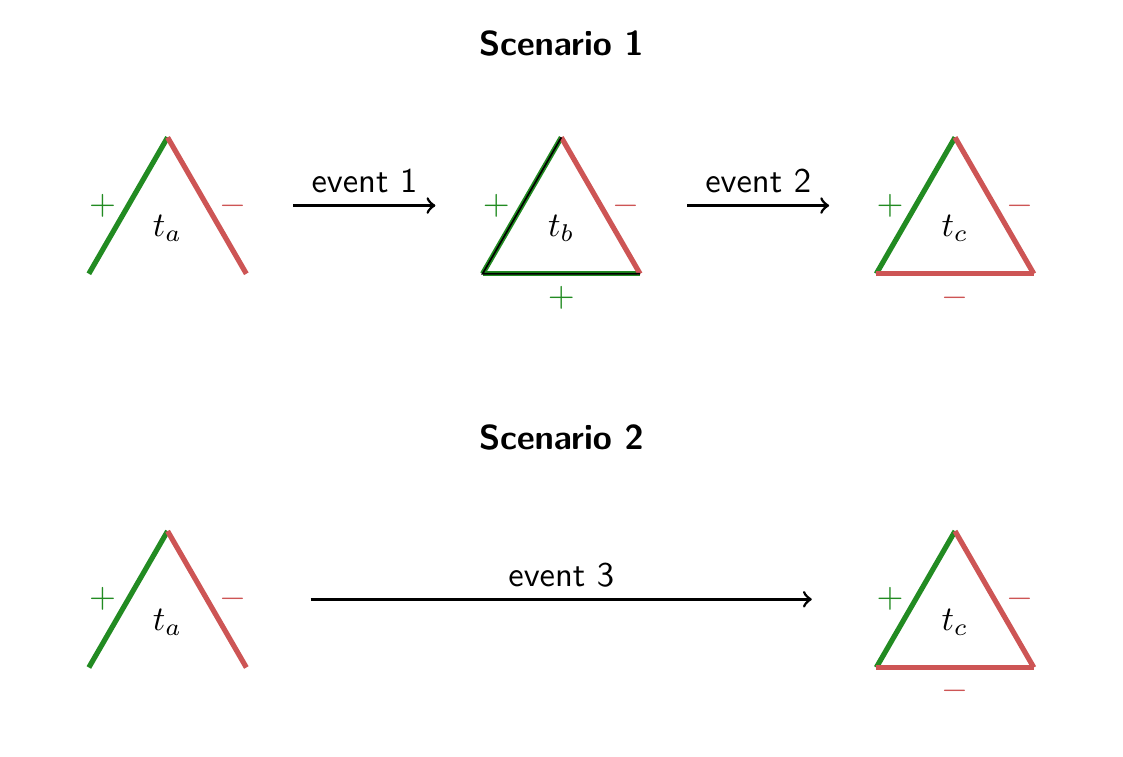}
  \caption{Dynamics of Balance. Two scenarios that start and end in the same balanced state, but the intermediary steps are different.
  }
  \label{fig:dynamic}
\end{figure}

%% file: Chapters/theory.tex
\section{Asymptotics of the Null Distribution}
\label{sec:theory}

One drawback of a permutation test is that the accompanying null distribution has no closed form. Monte~Carlo methods are required to approximate the distribution, and, for large graphs, the number of samples needed for an accurate approximation might be prohibitive. Additionally, the lack of a closed form makes it difficult to compare different null distributions.
We resolve both issues by showing that our statistic, namely, the number of unbalanced triangles, is asymptotically normal under the new null hypothesis, and therefore can be reasonably approximated by a normal distribution with known moments
\footnote{Ideally, we would have a Berry-Esseen type result here to quantify precisely the error in the normal approximation. We leave this for future work.}.
This is the main theoretical result of the paper (\Cref{thm:main}).



The distribution of sums of permutation-based random variables are notoriously difficult to analyze as the random variables themselves are no longer independent.
Thankfully, there is already a whole field dedicated to these types of results, known as combinatorial central limit theorems.
This field was founded by Hoeffding in his seminal paper \citep{Hoeffding:1951im}, in which he utilized the method of moments. Subsequent results have used variations of Stein's method to prove more general combinatorial CLTs (see \citet{Barbour:1989bt}).
While it is a relatively straightforward exercise to extend the results of \citet{Barbour:2005vl} to show asymptotic normality under the old null hypothesis, the additional structure introduced by our stratified permutation renders this approach infeasible.

Inspired by the results of \citet{Janson:2007dx}, we instead take a completely different approach.
The key insight is that the distribution of the stratified permutation model is a conditional distribution of a much simpler model, where the signs are distributed as independent (non-symmetric) Rademacher random variables, and the conditioning is on the number of negative ties in each stratum.

\subsection{Notation}
Fix a signed graph $G = (V, E, W)$. 
We can associate with $G$ a new graph-valued random variable, $G_{p} = (V, E, X)$, having the same support as $G$, but with a random weight vector $X$ comprising of independent (non-symmetric) Rademacher random variables. Concretely, the weight vector $X$ is given by $X_i \sim \operatorname{Rad}\left( 1 - p_{\varepsilon_i} \right)$, with $\P{ X_i = 1 } = 1 - p_{\varepsilon_i}, \P{ X_i = -1 } = p_{\varepsilon_i}$, where $p_{\varepsilon_i} = \frac{m_{\varepsilon_i}}{n_{\varepsilon_i}}$.
By construction, edges of the same embeddedness level will have the same probability of being negative under $G_{p}$. The statistic for $G_p$ corresponds to
\begin{align}
  V_{p} \coloneqq \sum_{(i,j,k) \in \triangle'} \I{X_{i} X_{j} X_{k} = -1}.
\end{align}
\subsection{Main Results}

The relation between $G_p$ and $G_\pi$ is given by the following Lemma.
\begin{lem}\label{lem:cond_dist}
  Denote by $M_l$ the number of negative ties with embeddedness $l$ in $G_p$. Then the distribution of $G_{\pi}$ is just the conditional distribution of $G_{p}$, conditional on the event $M=\bm{m}$, where $M=\left\{ M_l \right\}_{l=0}^{L}$ and  $\bm{m}=\left\{ m_l \right\}_{l=0}^{L}$. In particular,
  \begin{align}
    \mathcal{L}\left( U_{\pi} \right) = \mathcal{L}\left( V_p \,\mid\, M=\bm{m} \right).
  \end{align}
\end{lem}
The proof is deferred to \Cref{sec:appendix_lemmas}.
With respect to proving a CLT we must first define a sequence of graphs. Let $\left\{ G^{(n)} \right\}_{n=1}^{\infty}$ be a fixed sequence of signed graphs, indexed by the number of edges. Then, $G_\pi, G_p$ extend naturally to sequences of graph-valued random variables, $\left\{ G^{(n)}_{\pi} \right\}_{n=1}^{\infty}, \left\{ G^{(n)}_{p} \right\}_{n=1}^{\infty}$ with statistics $\left\{ U_\pi^{(n)} \right\}_{n=1}^{\infty}, \left\{ V_p^{(n)} \right\}_{n=1}^{\infty}$. To simplify notation, we will sometimes drop the index $n$. For instance, we still write $E_l$, $\pi$ and $p_l, m_l$, keeping in mind the dependence on $n$.

\Cref{lem:cond_dist} suggests that we can obtain a CLT for $U^{(n)}_\pi$ by proving one for $V^{(n)}_p$, which should be much easier by independence of edges in $G^{(n)}_p$. In general, however, weak convergence does not imply \emph{conditional} weak convergence. The crucial property that enables one to interchange conditioning and limits is \emph{stochastic monotonicity}, which we define below. Here, we use the partial order $x \leq y$ for vectors $x, y \in \mathbb{R}^{d}$ defined by $x_i \leq y_i$ for all $i$.
\begin{defn}\label{defn:stochastic_monot}
  Let $V \in \mathbb{R}^{q}$ and $M \in \mathbb{R}^{r}$ be random vectors. We say that $V$ is \emph{stochastically increasing} with respect to $M$ if the conditional distribution $\mathcal{L}(V \,\mid\, M = \bm{m})$ is increasing in $\bm{m}$. That is, if for any $v \in \mathbb{R}^{q}$ and $\bm{m}_1 \leq \bm{m}_2$, we have
  \begin{align}
    \P{V \leq v \,\mid\, M = \bm{m}_1} \geq \P{ V \leq v \,\mid\, M = \bm{m}_2}.
  \end{align}
  We say that $V$ is \emph{stochastically decreasing} with respect to $M$ if $-V$ is stochastically increasing with respect to $M$, and $V$ is \emph{stochastically monotone} with respect to $M$ if it is either stochastically increasing or decreasing with respect to $M$.
\end{defn}

Based on stochastic monotonicity it is indeed possible to transform weak convergence to conditional weak convergence. This beautiful result is known as \emph{Nerman's Theorem} (\citet{Nerman:1998fx,Janson:2007dx}).
\begin{thm}[Nerman's Theorem]\label{thm:nermans_theorem}
  Let $V^{(n)} \in \mathbb{R}^{q}$ and $M^{(n)} \in \mathbb{R}^{r}$ be random vectors such that $V^{(n)}$ is stochastically monotone with respect to $M^{(n)}$. Assume that
  \begin{align}
    (a_{n}^{-1} ( V^{(n)} - b_n), c_{n}^{-1}( M^{(n)} - d_n)) \xrightarrow{\mathcal{D}} (V,M)
  \end{align}
  for random vectors $V \in \mathbb{R}^{q}, M \in \mathbb{R}^{r}$ and $a_n, c_n > 0, b_n \in \mathbb{R}^{q}, d_n \in \mathbb{R}^{r}$. Let also $m_n \in \mathbb{R}^{r}$ be a sequence such that $c_n^{-1}(m_n - d_n) \to \xi \in \mathbb{R}^{r}$ and let $U^{(n)}$ be a random vector with distribution $\mathcal{L}(V^{(n)} \,\mid\, M^{(n)} = m_n)$. Suppose that $\xi$ is an interior point of the support of $M$ and that there exists a version of $m \mapsto \mathcal{L}(V \,\mid\, M = m)$ continuous at $m = \xi$ as a function of $m \in \mathbb{R}^{r}$ into $\mathcal{P}(\mathbb{R}^{q})$, the set of probability measures of $\mathbb{R}^{q}$. Then,
  \begin{align}
    a_{n}^{-1} \left( U^{(n)} - b_n \right) \xrightarrow{\mathcal{D}} \mathcal{L}(V \,\mid\, M = m).
  \end{align}
\end{thm}

\begin{cor}[Corollary 2.5 of \cite{Janson:2007dx}]\label{cor:linear_op}
  We can replace the assumption above that $V^{(n)}$ is stochastically monotone with respect to $M^{(n)}$ by the assumption that $H V^{(n)}$ is stochastically monotone with respect to $M^{(n)}$ for some invertible linear operator $H$ on $\mathbb{R}^{d}$.
\end{cor}
\begin{proof}
  Apply the theorem to $(H V^{(n)}, M^{(n)})$ with $V$ and $b_n$ replaced by $HX$ and $H b_n$, respectively. The result follows by applying $H^{-1}$.
\end{proof}
The number of unbalanced triangles is not stochastically monotone with respect to the number of negative ties in each stratum, so we cannot apply this theorem directly to $V^{(n)}_{p}$. It is not hard to see, however, that the counts of triangles having at least $\alpha=1,2,3$ negative ties satisfies stochastic monotonicity. 
Let $T_\alpha^{(n)}$ denote the number of triangles in $G_p^{(n)}$ with $\alpha$ negative ties.
Then, we have the following result:
\begin{lem}\label{lem:stoch_mono}
The vector $H T = (T_{3}, T_{3} + T_{2}, T_{3} + T_{2} + T_{1})$ is stochastically increasing with respect to $M$, where $T\coloneqq(T_{1}, T_{2}, T_{3})$, $M\coloneqq(M_{0}, \ldots, M_{L})$ and the matrix $H$, given by $Hx=(x_3,x_3+x_2,x_3+x_2+x_1)$ for $x\in \mathbb{R}^{3}$, is invertible.
\end{lem}
The proof is deferred to \Cref{sec:appendix_lemmas}.
Since $H$ is an invertible linear operator, and $V^{(n)}_{p} = T_{1}^{(n)} + T_{3}^{(n)}$, \Cref{cor:linear_op} and \Cref{lem:stoch_mono} together shows that it is sufficient to prove a CLT for $(T^{(n)}, M^{(n)})$.
This requires a few mild assumptions.


\begin{ass}\label{ass:max_emb}
  The embeddedness level of negative ties in $\left\{ G^{(n)} \right\}_{n = 1}^{\infty}$ is bounded from above by some $L_{-} < \infty$.
\end{ass}

Our first assumption is predominantly a technical one, as Nerman's Theorem does not apply when $L^{(n)}$, the dimension of $M^{(n)}$, is unbounded. Note that, under the current definition of $L^{(n)}$, which we recall is the largest embeddedness value in the graph $G^{(n)}$, we would require an upper bound on the embeddedness level of all ties (not just negative ties). But in fact it suffices to define $M^{(n)}$ up to the largest embeddedness level for negative ties (say $L_{-}^{(n)}$), as the permutations above $L_{-}^{(n)}$ (with no negative ties) would be degenerate.

Moreover, in practice, the embeddedness level does not grow with the size of the graph.
For instance, across the 32 village networks in our dataset (with the number of edges ($n$) ranging from 54 to 1109), the largest embeddedness for negative ties was 5, while the largest embeddedness for positive ties was 13.
In fact, the largest graph (with 1109 edges) had a maximum embeddedness for negative ties of only 2.

\begin{ass}\label{ass:emb2}
  We require $\sup_{i \in E} \varepsilon_i^2 = o(n).$
\end{ass}

\begin{ass}\label{ass:all_embed_limit}
  For $l_{1}, l_{2}, l_{3} = 0, \ldots, L_{-}$ define $E_{l_1,l_2,l_3} \coloneqq E_{l_1}\times E_{l_2}\times E_{l_3}$.
  We assume that all partial sums
  \begin{align}
    \frac{1}{n} \sum_{(i,j,k) \in E_{l_1,l_2,l_3}} \triangle_{ijk},
    \qquad
    \frac{1}{n} \sum_{i \in E_{l_1}}
    \left(\sum_{(j,k) \in E_{l_2,l_3}}
    \triangle_{ijk}\right)
    \left(\sum_{(j',k') \in E_{l_4,l_5}}
    \triangle_{ij'k'}\right),
  \end{align}
  converge to a limit as $n \to \infty$, where $\triangle_{ijk} \coloneqq \I{(i,j,k) \in \triangle}$.
  Separately, we require $p_{l} = \dfrac{m_l}{n_l}$ to also converge to a limit as $n \to \infty$.
\end{ass}
This assumption is needed to ensure that the limiting covariance terms exist.
  We require such a condition because our statistic is intimately related to the structure of the sequence $\left\{ G^{(n)} \right\}_{n=1}^{\infty}$, which is fixed. This contrasts with other random graph models where the support of the graph is the main modeling task.
  A simple consequence of \Cref{ass:all_embed_limit} is that $\dfrac{n_l}{n}$ also converges, as
\begin{align}
  \frac{1}{2l}
  \sum_{l_2, l_3 = 0}^{L_{-}} \frac{1}{n} \sum_{(i,j,k) \in E_{l,l_2,l_3}} \triangle_{ijk}
  &= \frac{2}{2nl}  \sum_{i \in E_{l}} \varepsilon_i =  \frac{1}{nl} n_l l = \frac{n_l}{n} \label{eqn:n_l}
\end{align}
shows that $\dfrac{n_l}{n}$ is a finite, linear combination of terms that converge.
Consider therefore the normalized random variables $\widetilde{T}^{(n)} \coloneqq \frac{1}{\sqrt{n}}{\left(T^{(n)} - \E{T^{(n)}} \right)}$, $\widetilde{M}^{(n)} \coloneqq \frac{1}{\sqrt{n}}{\left(M^{(n)} - \E{M^{(n)}} \right)}$, where we recall that the variables $T^{(n)} = (T^{(n)}_{1}, T^{(n)}_{2}, T^{(n)}_{3})$ relate to the independent graph model $G_{p}^{(n)}$.

\begin{prop}\label{prop:clt_radamacher}
 Under \Cref{ass:max_emb,ass:emb2,ass:all_embed_limit}, we have for $n \to \infty$,
  \begin{align}
    \left( \widetilde{T}^{(n)}, \widetilde{M}^{(n)} \right) \xrightarrow{\mathcal{D}} \left( \widetilde{T}, \widetilde{M} \right)
    \sim \mathcal{N}(0, \Sigma),
  \end{align}
  where $\Sigma$ is given in \cref{eqn:sigma3} of \Cref{sec:appendix_proof}.
\end{prop}

\begin{rmk}
  As $\Sigma$ is an asymptotic variance, in practice, it is enough to approximate $\Sigma$ by the covariance matrix $\Sigma^{(n)}$ of $(\widetilde{T}^{(n)}, \widetilde{M}^{(n)})$, as given in \cref{eqn:sigma2} of \Cref{sec:appendix_proof}, using \cref{eqn:cov_t1,eqn:cov_t2,eqn:cov_t3,eqn:cov_t_m,eqn:cov_m}.
\end{rmk}


The proof is deferred to \Cref{sec:appendix_proof}.
We are now ready to state and prove our main theorem.
\begin{thm}\label{thm:main}
  Grant \Cref{ass:max_emb,ass:emb2,ass:all_embed_limit}.
  Under the new null hypothesis, $H_{0}^{\pi}$,
  we have for $n \to \infty$ that the normalized count of unbalanced triangles has a limiting normal distribution:
  \begin{align}
    n^{-1/2}\left( U^{(n)}_{\pi} - \E{ V^{(n)}_p } \right) \xrightarrow{\mathcal{D}} \mathcal{N}(0, \sigma^2_{u}),
  \end{align}
  with $\sigma^2_{u} = \Sigma^{s}_{1,1} + \Sigma^{s}_{3,3} + 2 \Sigma^{s}_{1,3}$ and $\Sigma^{s} = \Sigma_{\widetilde{T},\widetilde{T}} - \Sigma_{\widetilde{T},\widetilde{M}} \Sigma_{\widetilde{M},\widetilde{M}}^{-1} \Sigma_{\widetilde{T},\widetilde{M}}^{\top}$, where $\Sigma_{\widetilde{T},\widetilde{T}}$, $\Sigma_{\widetilde{M},\widetilde{M}}$ and $\Sigma_{\widetilde{T},\widetilde{M}}$ are the covariance matrices of $\widetilde{T}, \widetilde{M}$ from \Cref{prop:clt_radamacher}.
\end{thm}

\begin{proof}
  \Cref{prop:clt_radamacher} gives us joint asymptotic normality of $(\widetilde{T}^{(n)}, \widetilde{M}^{(n)})$. 
  By stochastic monotonicity of $(T^{(n)}, M^{(n)})$ (\Cref{lem:stoch_mono}), we can apply Corollary 2.5 of \citet{Janson:2007dx} (a variation of Nerman's Theorem) to get
  \begin{align}
    \widetilde{S}^{(n)} \coloneqq \mathcal{L}\left( \widetilde{T}^{(n)} \,\mid\, \widetilde{M}^{(n)} = 0 \right) \xrightarrow{\mathcal{D}} \mathcal{L}\left( \widetilde{T} \,\mid\, \widetilde{M} = 0 \right).
  \end{align}
  Since $(\widetilde{T}, \widetilde{M})$ has a joint normal distribution, it is well known that the conditional distribution is also normally distributed. In other words, $\widetilde{S}^{(n)} \xrightarrow{\mathcal{D}} \widetilde{S} \sim N(0, \Sigma^{s})$ and, in particular,
  \begin{align}
    n^{-1/2}
    \left( U^{(n)}_{\pi} - \E{V^{(n)}_{p}} \right) = \widetilde{S}^{(n)}_{1} + \widetilde{S}^{(n)}_{3} \xrightarrow{\mathcal{D}} \mathcal{N}(0, \sigma^2_{u}).
  \end{align}
\end{proof}


\begin{rmk}
  Note that $\widetilde{M}_l^{(n)}$ is degenerate if $\frac{m_l}{n} \to 0$.
  To ensure that we don't invert a degenerate covariance matrix, we can simply remove the degenerate $\widetilde{M}_l^{(n)}$ from $\widetilde{M}^{(n)}$. The result still holds, as stochastic monotonicity is maintained with the smaller $\widetilde{M}^{(n)}$.  
\end{rmk}

  From the form of the limiting distribution, it is clear that the asymptotic means of $U^{(n)}_{\pi}$ and $V^{(n)}_{p}$ coincide (though the variances differ).
  This suggests that one could save a lot of trouble by adopting the independent model $G_{p}$ instead of using the permutation model $G_{\pi}$, with little difference in results besides some \hl{inevitable increase} in variance.
  This is indeed true when the size of the graph is very large, but for the rest, like the graphs we collected in our dataset, the discrepancy between the two models is nontrivial.
  In particular, due to the small size of our graphs, and the small number of negative ties observed, empirically we find that the graphs generated under the independent model will often have no negative ties, rendering the task of measuring balance moot.

As a byproduct, we also obtain asymptotic normality for $U^{(n)}_{\tau}$ in the old model.
\begin{cor}\label{cor:an_uniform}
  Asymptotic normality of $U^{(n)}_{\tau}$, the statistic under \Cref{modl:old}, follows from \Cref{thm:main}, by replacing the vector $M^{(n)}$ by the sum $\sum_{l = 0}^{L_{-}} M_{l}^{(n)}$, and letting $p_{\varepsilon_i} = \frac{m}{n}$ for all $i \in E$.
\end{cor}

\begin{rmk}
  The proof of \Cref{lem:cond_dist} requires only that the probabilities for an embeddedness level are identical (i.e. $p_i = p_j \text{ if } \varepsilon_i = \varepsilon_j$). In particular, this means that one could also choose $p_{\varepsilon_i} = \frac{m}{n}$ for all $i \in E$. By doing so, one could start with the same independent graph random variable, and then we can recover both permutation random variables, depending on if we were to condition on each embeddedness level separately ($M = \left\{ m_i \right\}_{i=1}^{L_{-}}$) or on the total number of negative ties ($M = \sum_{i = 1}^{L_{-}} m_i$).

  However, if we were to use the random graph with uniform probability $\frac{m}{n}$ of being negative, then the event that we condition on, $\left\{ M^{(n)} = \bm{m}^{(n)} \right\}$, is not at the mean of the distribution (recall that we are forced to condition on the empirical counts of the embeddedness levels). Thus the statement of our theorem becomes degenerate as, in the limit, $\P{ M = \bm{m}} = 0$.
\end{rmk}

\subsection{Comparison of the two Models}
\label{sub:comparing_both}



We have argued in \Cref{sec:method} that a graph can be considered balance-free, if the sign of an edge is independent of its position in the graph, relative to its embeddedness. Among other issues, ignoring embeddedness means treating all edge labels as exchangeable which is not true for social networks. We concluded that reasonably balance-free graphs can be generated with respect to the restricted random permutation $\pi$. Building on this idea, in this section we will construct a hierarchy of generative models producing balance-free graphs and show, if these graphs are assumed as null models, that the old test with respect to the critical values $c_{\alpha,\tau}^{(n)}$ derived from \Cref{modl:old} has Type-I error converging to 1 as $n \to \infty$, while our test with critical values $c_{\alpha,\pi}^{(n)}$ from \Cref{modl:new}  has Type-I error matching the specified $\alpha$. This proves formally that the new test is more conservative than the old one. \Cref{sec:simulations} will show, on the other hand, that the new test still detects balance, if it exists. 

Arguing by the normal approximations in \Cref{thm:main} and \Cref{cor:an_uniform}, comparing the critical values essentially reduces to comparing the means and variances of Gaussians. This yields the following result.
\begin{prop}\label{lem:old_test_fails}
Grant \Cref{ass:max_emb,ass:emb2,ass:all_embed_limit} and assume that there exists a constant $c>0$ such that for large $n$
\begin{align}
    (1-p)^2 \frac{\sum_{i \in E} \varepsilon_i}{n} - \frac{ \sum_{i \in E^{-}} \varepsilon_i }{m} > cm^{1/2}\log{m}, \label{eqn:difference_means}
\end{align} 
where $E^{-}\subset E$ is the set of negative ties in $G^{(n)}$. 
Then $\P{U^{(n)}_{\pi}\leq c_{\alpha,\pi}^{(n)}}\rightarrow \alpha$, while $\P{U^{(n)}_{\pi}\leq c_{\alpha,\tau}^{(n)}} \rightarrow 1$ as $n \to \infty$, i.e. the Type-I error of the original test applied to graphs generated by $G^{(n)}_{\pi}$ converges to 1.
\end{prop}

The proof is deferred to \Cref{sec:appendix_lemmas}. Observe that \cref{eqn:difference_means} requires the difference in the average embeddedness values of all ties versus just negative ties to have enough separation. These two terms essentially reflect the means of the two distributions.


Now, let us generalize the model. Instead of using the stratified permutation, the same results hold if we replace $G_{\pi}^{(n)}$ with the Rademacher model $G_p^{(n)}$.
Recall that the new test derived from \Cref{modl:new} (and the respective critical value) depends only on $\bm{m}^{(n)}$ (provided the support of the graph is fixed). Now, graphs generated from $G_{\pi}^{(n)}$ all have the same value of $\bm{m}^{(n)}$ so they all share the same critical value. However, when we generalize to $G_{p}^{(n)}$, then the critical value will be a function of the $\bm{m}^{(n)}$. Let us denote the critical value by $c_{\alpha,\pi}(\bm{m}^{(n)})$. Then, 
\begin{align}
  \P{ V_{p}^{(n)} \leq c_{\alpha,\pi}(\bm{m}^{(n)})}
  &= \sum_{\bm{m}^{(n)}} \P{ V_{p}^{(n)} \leq c_{\alpha,\pi}(\bm{m}^{(n)}) \,\mid\, M^{(n)} = \bm{m}^{(n)}} \P{ M^{(n)} = \bm{m}^{(n)} }
  \\&= \sum_{\bm{m}^{(n)}} \P{ U_{\pi}^{(n)} \leq c_{\alpha,\pi}(\bm{m}^{(n)})} \P{ M^{(n)} = \bm{m}^{(n)} }  \label{eqn:crit_error}
  \\&= \sum_{\bm{m}^{(n)}} \alpha \cdot \P{ M^{(n)} = \bm{m}^{(n)} }.
\end{align}
On the other hand, considering the critical value for the old test, \cref{eqn:crit_error} would be
\begin{align}
  \P{ V_{p}^{(n)} \leq c_{\alpha, \tau}(\bm{m}^{(n)}) }
  &= \sum_{\bm{m}^{(n)}} \P{ U_{\pi}^{(n)} \leq c_{\alpha, \tau}(\bm{m}^{(n)})} \P{ M^{(n)} = \bm{m}^{(n)} }.
\end{align}
By \Cref{lem:old_test_fails}, we have that for each term, $\P{ U_{\pi}^{(n)} \leq c_{\alpha, \tau}(\bm{m}^{(n)})} \to 1$ as $n \to \infty$. Thus, since $\sum_{\bm{m}^{(n)}} \P{ M^{(n)} = \bm{m}^{(n)} } = 1$, we have that $ \P{ V_{p}^{(n)} \leq c_{\alpha, \tau}(\bm{m}^{(n)}) } \to 1$ as $n \to \infty$ as required.

Finally, we can generalize the balance-free graph model $G_p^{(n)}$ even further to consider general graphs (that is, no longer confined to a fixed support). This can be achieved by simply adding an additional step to the generative model: first randomly draw a unsigned graph $G$ (the support) from some distribution over all possible graphs on $N$ vertices, then draw from $G_p^{(n)}$. A similar argument to the one above gives the result.

%% file: Chapters/simulation.tex
\section{Simulations}
\label{sec:simulations}

In this section, we first provide empirical verification of our theoretical results showing asymptotic normality of the statistic under \Cref{modl:new}, the stratified permutation null model.
The rest of the section is then dedicated to comparing the performance of the tests under generative models of graphs without balance $(H_0)$ as well as graphs with balance $(H_1)$.
Under models of balance-free graphs,
we show that our test is non-significant, while the old test exhibits spurious significance.
This corroborates with our analysis of the Type-I error in \Cref{sub:comparing_both}.
Finally, we simulate graphs with balance, and show that our test maintains the same level of statistical power as the old test. Since our test is generally more conservative, this is the best result possible.



\subsection{Asymptotic Distribution}
We generated three progressively larger graphs from a Watts-Strogatz model \citep{Watts:1998db}, with parameters given in \Cref{tbl:exact_sw}, and then randomly assigned a fraction of these edges to be negative.
The Watts--Strogatz model, with parameters ($d, n, k, p$), has the following generative mechanism. Form a $d$-dimensional lattice with $n$ nodes per dimension. Then, connect two vertices together if the number of hops on the original lattice between them is at most $k$.
Finally, iterating over each edge, rewire each end with probability $p$.
We ran $10^4$ Monte Carlo simulations to calculate the empirical distribution of the number of unbalanced triangles under $H^{\pi}_{0}$.
These empirical distributions are shown in \Cref{fig:exact}, where we see a clear trend towards convergence to normality.
\begin{figure}[h]
  \centering
  \includegraphics[width=.8\linewidth]{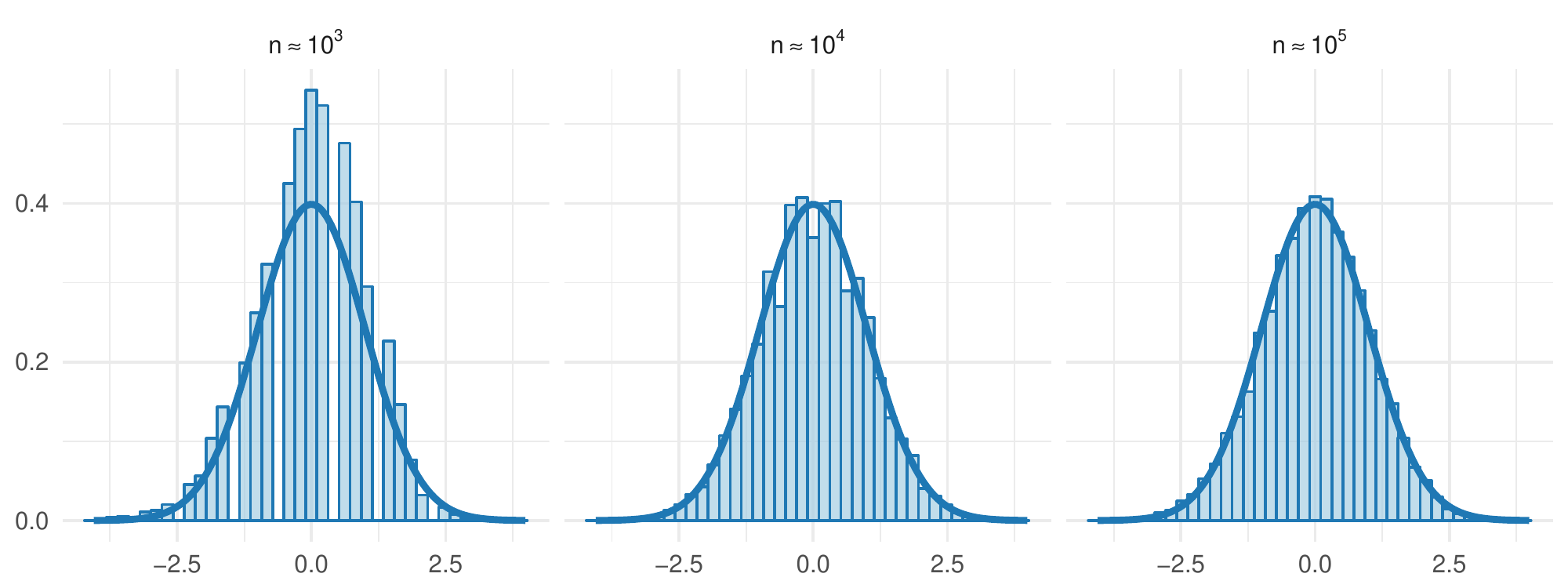}
  \caption{Histogram of the sample null distribution given the three base graphs, with the density curve of a standard Normal distribution superimposed.}
  \label{fig:exact}
\end{figure}
\begin{table}[h]
  \small
  \centering
  \begin{tabular}{lrrrrr}
    \toprule
    & \emph{d} & \emph{n} & \emph{k} & $p$ \\
    \midrule
    $n \approx 10^3$ & 3 & 4 & 3 & 0.2 \\
    $n \approx 10^4$ & 3 & 6 & 3 & 0.2 \\
    $n \approx 10^5$ & 5 & 10 & 5 & 0.1 \\
    \bottomrule
  \end{tabular}
  \caption{Parameters used in each model of the Watts-Strogatz graph.}
  \label{tbl:exact_sw}
\end{table}

\subsection{Performance Comparisons}

\subsubsection{Comparison under \texorpdfstring{$H_0$}{H0}}

To ensure a fair comparison, we chose different models of balance-free graphs from the permutation based models our tests use.
Our generating process is to generate positive and negative subgraphs independently (on the same vertex set), and then combine the graphs together, in such a way that the negative subgraph takes precedence. By this procedure, the positive and negative edge sets will be independent, which by definition produces a balance-free graph.

The negative subgraph will be simply drawn from an \ER model (turning all edges to negative ones).
For the positive subgraph model, we would like to use a graph model that is a faithful representation of real-life social networks. Here we present two choices:

\textsc{Choice 1: Small-world}.
In this simulation model, we generate the positive part of the social network from the Watts-Strogatz model \citep{Watts:1998db}.

The model parameters we used in this simulation are \emph{d} $=1$, \emph{n} $=100$, \emph{k} $=2$, and we varied the rewiring probabilities from $p=0.1, 0.2, 0.3$.
We draw $10^3$ samples each from three instances of the above generative model, and compare the $p$-values the two tests produce, the histograms of which are shown in \Cref{fig:hist_sw}.
For a rewiring probability of $p = 0.1$, we find that the old test is rejecting all $10^3$ graphs, at a significance value of 0\%. On the other hand, the new test rejects the null only 10\% of the time.
As we increase $p$, the graph becomes more like an \ER graph, and the two tests tend towards uniformity, though from opposite directions. Crucially, in all three instances of $p$, our new test is conservative about rejecting the null, while the old test is not.
\begin{figure}[h]
  \centering
  \includegraphics[width=.8\linewidth]{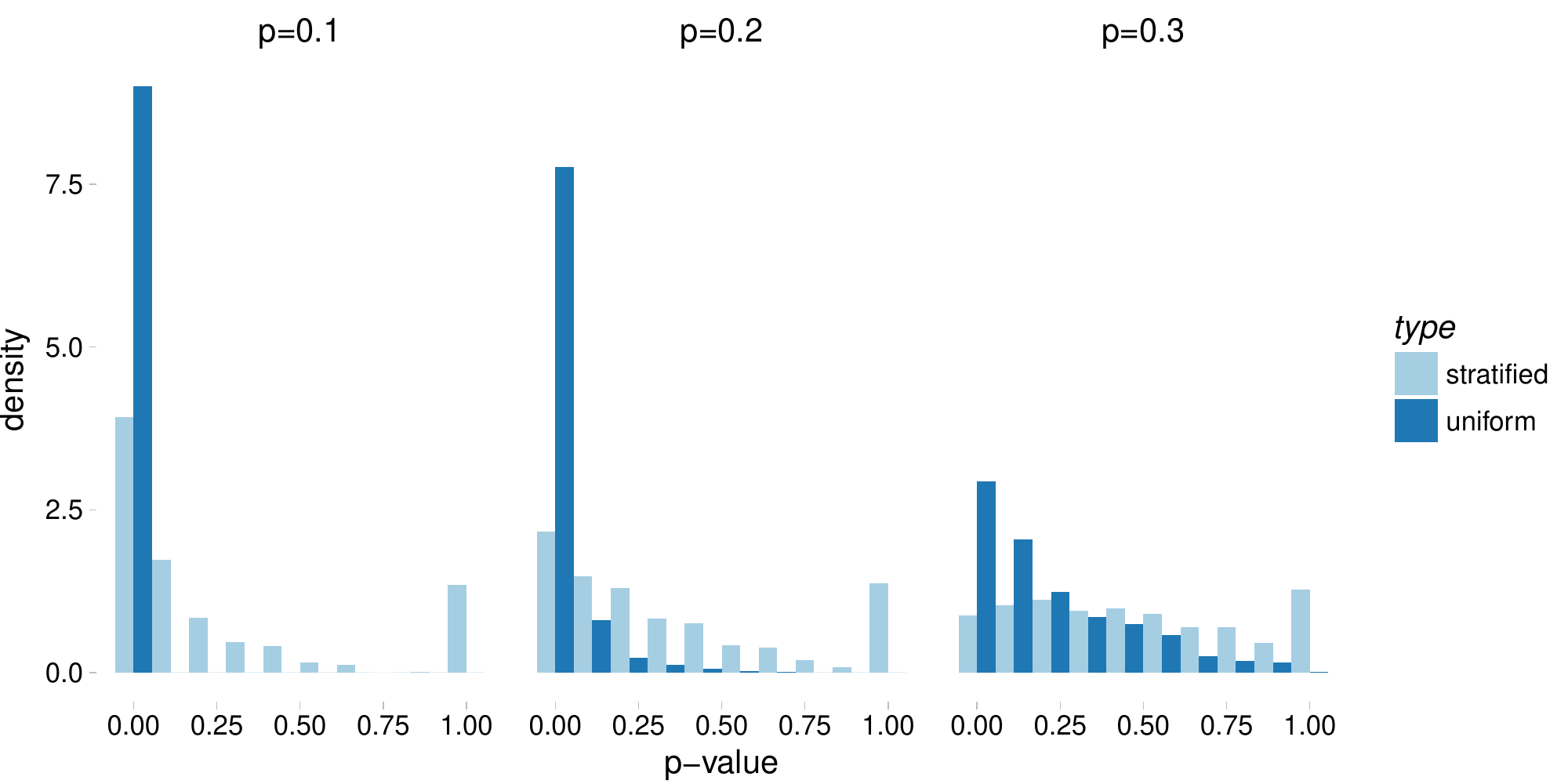}
  \caption{Histogram of the $p$-value for differing parameters of the Watts-Strogatz model}
  \label{fig:hist_sw}
\end{figure}

\textsc{Choice 2: Real Data}.
A simple way to ensure that the positive subgraph possesses the features we see in real data is
to just use real data.
We simply removed the negative edges from the social networks we collected, leaving a positive subgraph.
The negative \ER graph is then added.
Selecting three representative villages, we draw $10^3$ samples each from the three resulting models and compare the $p$-values we get from the two tests. The histograms are shown in \Cref{fig:hist_village}.
We see a similar story across the villages, with the old test rejecting very often, while the new test almost never rejects. \todo{maybe some conclusion here?}

\begin{figure}[h]
  \centering
  \includegraphics[width=.8\linewidth]{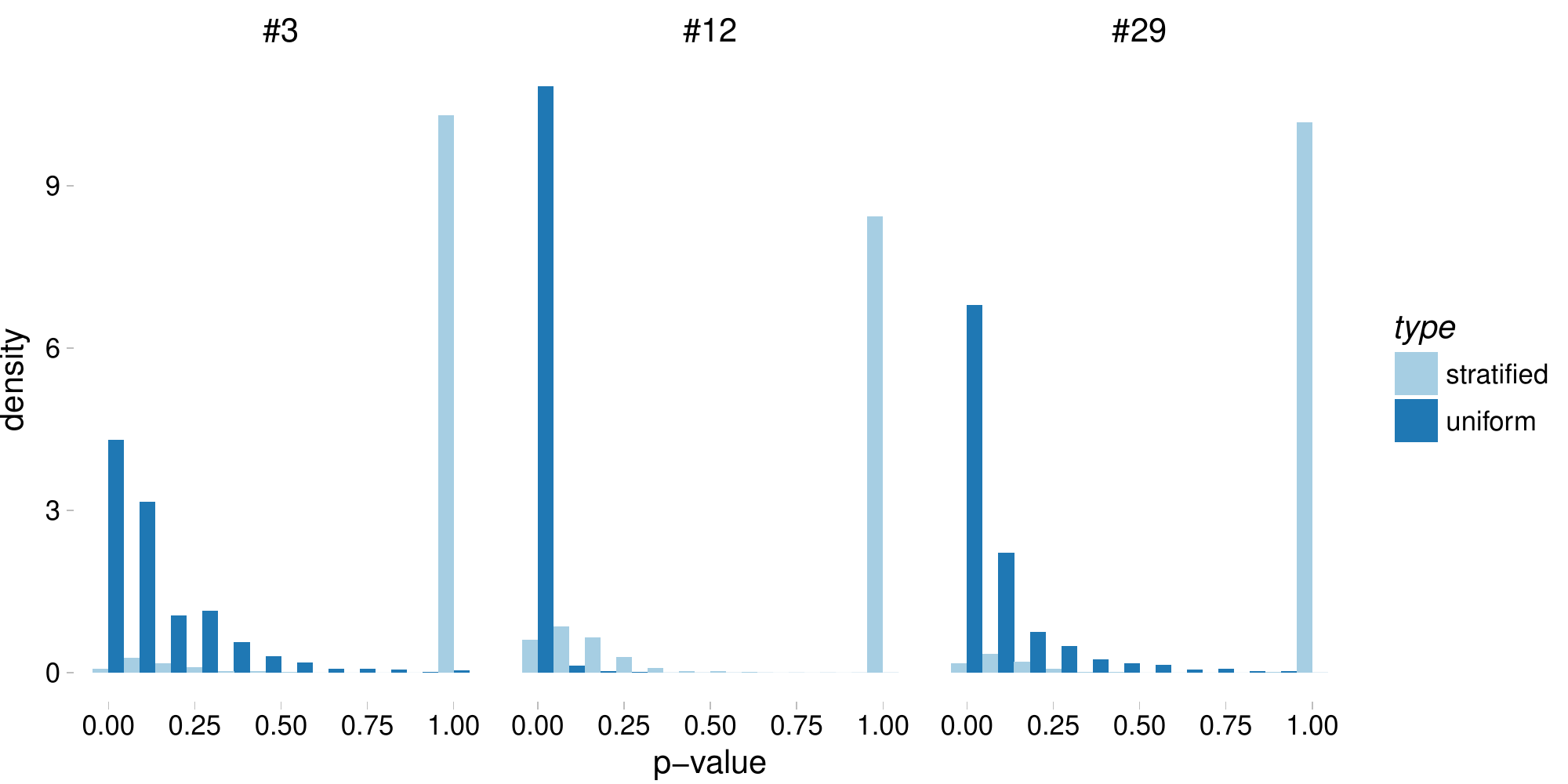}
  \caption{Histogram of the $p$-values for a sample of the village networks (\#2, \#12, \#29).}
  \label{fig:hist_village}
\end{figure}

\subsubsection{Comparison under \texorpdfstring{$H_1$}{H1}}


In order to derive a model of a graph with balance,
it will be informative to revisit the original definition of balance in \citep{Heider:1946vz}. There,
a graph is balanced if and only if it can be decomposed into two ``communities'' such that positives ties are within communities and negative ties are between.
Define a signed stochastic blockmodel as the combination of two stochastic blockmodels (SBM) over the same community structure, one for each sign\footnote{Clashes between two edges are deemed as void.}.
For notational convenience, we shall only consider the symmetric 2-community model. Let the parameters of the two SBMs be given by
\begin{align}
  B^{+} = 
  \begin{bmatrix}
    p^{+} & q^{+} \\
    q^{+} & p^{+}
  \end{bmatrix} &&
  B^{-} =
  \begin{bmatrix}
    p^{-} & q^{-} \\
    q^{-} & p^{-}
  \end{bmatrix} &
\end{align}
Then, balance corresponds to $q^{+} = p^{-} = 0$. On the other hand, the graph is balance-free when $p^{+} = q^{+}$ and $p^{-} = q^{-}$. A natural solution to modeling a graph with balance is therefore one where the parameteres are somewhere between the two degenerate solutions.

We ran the model with three different sets of parameters (see \Cref{tbl:sim_sbm_params}), and the histograms of $p$-values are shown in \Cref{fig:hist_sbm}. The distribution of $p$-values in the presence of balance are the same across the two methods, which demonstrates that our method doesn't lose out on statistical power. Compared to Model 1 and 2, the level of balance in Model 3 is much lower, and accordingly, the tests are more uncertain, leading to a more uniform distribution.

\begin{table}[h]
  \small
  \centering
  \begin{tabular}{lrrrrr}
    \toprule
    & $n$ & $p^{+}$ & $q^{+}$ & $p^{-}$ & $q^{-}$ \\
    \midrule
    Model 1 & 100 & 0.4 & 0.1 & 0.03 & 0.1 \\
    Model 2 & 50 & 0.3 & 0 & 0 & 0.3 \\
    Model 3 & 50 & 0.3 & 0.2 & 0.2 & 0.3 \\
    \bottomrule
  \end{tabular}
  \caption{Parameters used in each model of the signed SBM.}
  \label{tbl:sim_sbm_params}
\end{table}

\begin{figure}[h]
  \centering
  \includegraphics[width=.8\linewidth]{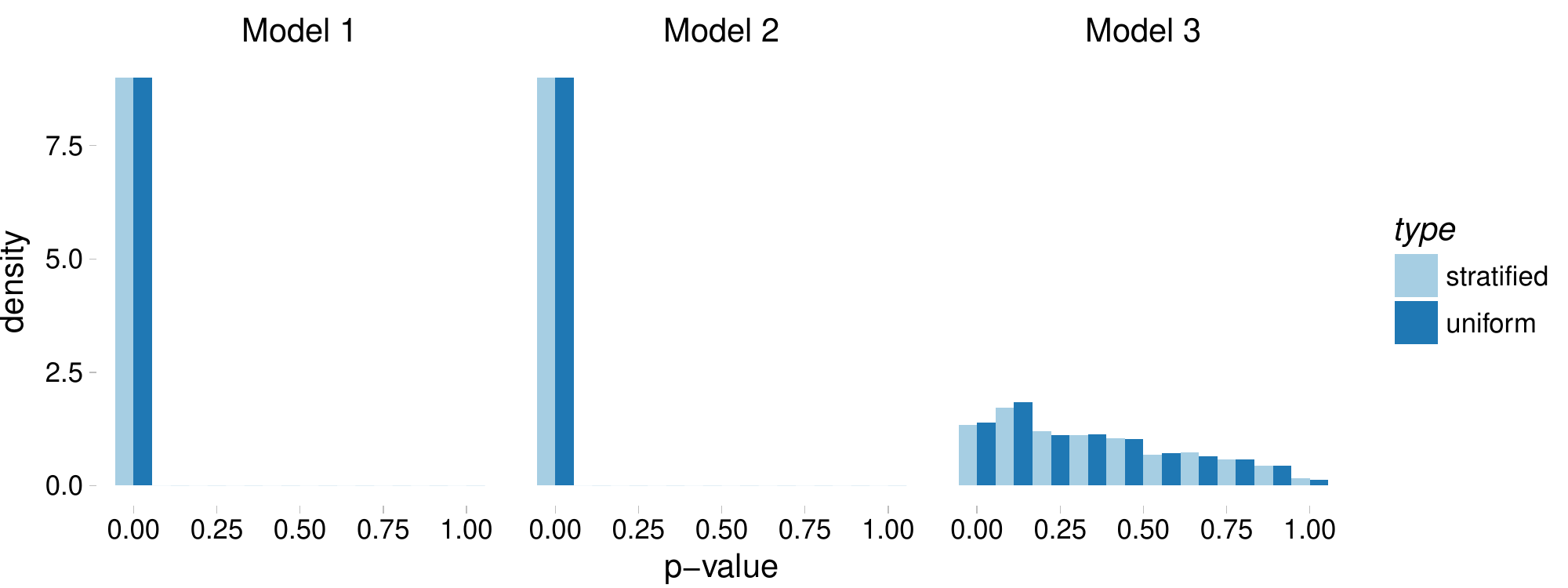}
  \caption{Histogram of the $p$-values from testing for balance in the signed SBM.}
  \label{fig:hist_sbm}
\end{figure}

%% file: Chapters/honduras.tex

\section{Case Study: Villages in Rural Honduras}
\label{sec:data}


Signed networks were first studied in the context of international relations: \citet{harary1961structural,moore1979structural} studied several international conflicts, from the Middle Eastern Suez Crisis of 1956 to the Indo-Pakistani War of 1971, while \citet{Antal:2006cia} studied the evolution of alliances during World War I.
Similarly, though at a scale considerably smaller than nation-states, \citet{hage1973graph,Hage1984} studied the New Guinea tribe warfare relations.
The nodes in these signed graphs are all collective entities, as opposed to individuals, and the ties themselves are primarily focused on warfare, rather than simple social engagements.

The actual sociocentric mapping of negative ties in parallel with positive ties in social networks is uncommon \citep{Everett:2014cb}.
\citet{rawlings2017structural} examined antagonistic ties in 129 people in a sample of 31 urban communes in the USA from the 1970's, and the classic \citet{sampson1969crisis} study of 18 novitiate monks collected information about members of the group who were disliked. Studies have also mapped helpful and adversarial relationships in small groups in classrooms \citep{huitsing2012bullying,mouttapa2004social} or workplaces \citep{xia2009exploring,Labianca:2006hk}. Other work has examined the networks formed by wild mammals \citep{Ilany:2013ig,lea2010heritable}.


A more recent source of signed networks are those derived from online websites: \citet{Kunegis:2009ks} considered the endorsement graph on the web forum Slashdot; \citet{Guha:2004hs} analyzed the trust/distrust network on the online review website Epinions; and \citet{burke2008mopping} looked at the public voting records for Wikipedia admin candidates.
Online networks have the advantage of scale (the above networks are on the order of $10^5$ nodes), but their artifically constructed notions of valence make them much less generalizable to real world settings.



\subsection{Rural Social Networks Study}


In the summer of 2010, the Rural Social Networks Study (RSNS) collected data on the social networks of around 5000 respondents spread across 32 rural villages in the La Union, Lempira region of Honduras \citep{kim2015social}.
The villages are geographically close to one another but there is little between-village communication.
The study gave a small survey to about 87 percent of the respondents in each village before playing a set of economic games.
The survey included a small number of demographic controls and concentrated on ``name generators'' to collect data on the social networks of each village.

For the name generators, the RSNS used a photographic census of all residents, coupled with bespoke software (a much revised version of which, known as \texttt{Trellis}, is available online at \url{http://trellis.yale.edu/}).  This software program was created to use photographs for the identification of ``alters'' to increase accuracy and efficiency of collecting social network data in the field. The photos help solve the name similarity problem. For instance, in one village, there were 16 Maria Hernandezs.



The name generators primarily focused on receiving strong affective relationships of each respondent: kinship, best friends, and matrimony. But the study also asked a question about ``general dislike''. In all, about 10 percent responded with alters to the negative affective name generator.
One of the weaknesses of self-reported antagonistic ties is that people display a reticence to speak negatively of other people in their communities to strangers.
We assume that the actual levels of animosity in these communities are higher than we are able to report because of this social desirability.
Name generators, by their very nature, produce data that is directional. We shall work with the symmetrized undirected version of this dataset\footnote{This dataset is hosted at our Lab's portal (\url{http://humannaturelab.net/}).}.


\subsection{The Behaviour of Negative Ties}

We begin by considering the negative ties as a separate entity, and compare their behavior to positive ties.
\Cref{fig:villages} shows representative village subgraphs of each type.
Clearly there is no mistaking the two.
The positive subgraphs are as expected, conforming to our established ideas of social networks.
Meanwhile, the negative subgraphs are extremely sparse, with most ties being isolated, and almost all the components are trees (i.e. no cycles).
There are instances of cycles, but they are exceedingly rare.
In fact, as hinted earlier, the negative subgraphs bear a striking resemblance instead to sparse random graphs, with their locally tree-like structure.


\begin{figure}[h]
  \centering
  \begin{subfigure}[b]{.49\textwidth}
    \includegraphics[width=.9\textwidth]{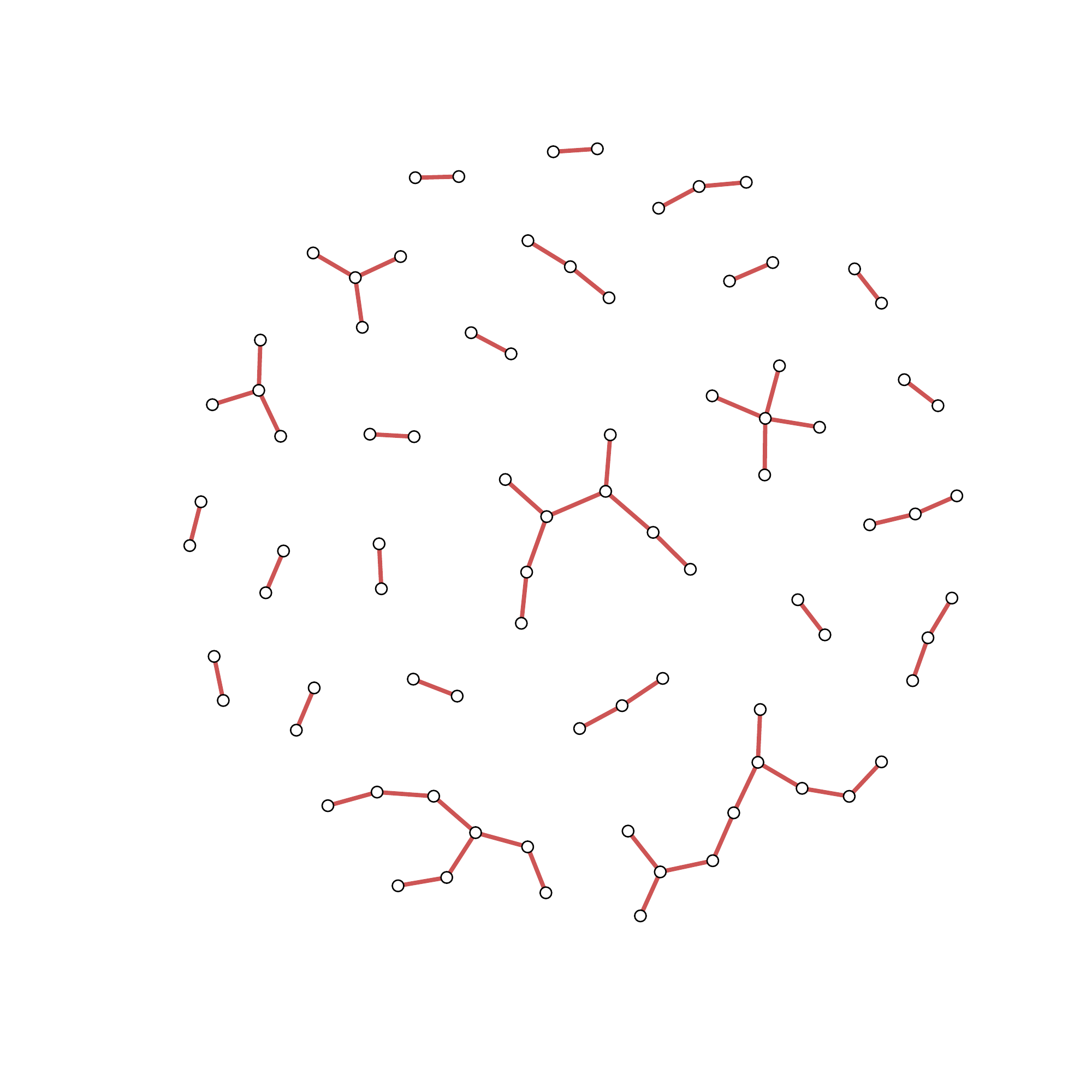}
    \vspace*{-2em}
    \subcaption{Village \# 11: Negative}
    \label{subfig:n11}
  \end{subfigure}
  \begin{subfigure}[b]{.49\textwidth}
    \includegraphics[width=.9\textwidth]{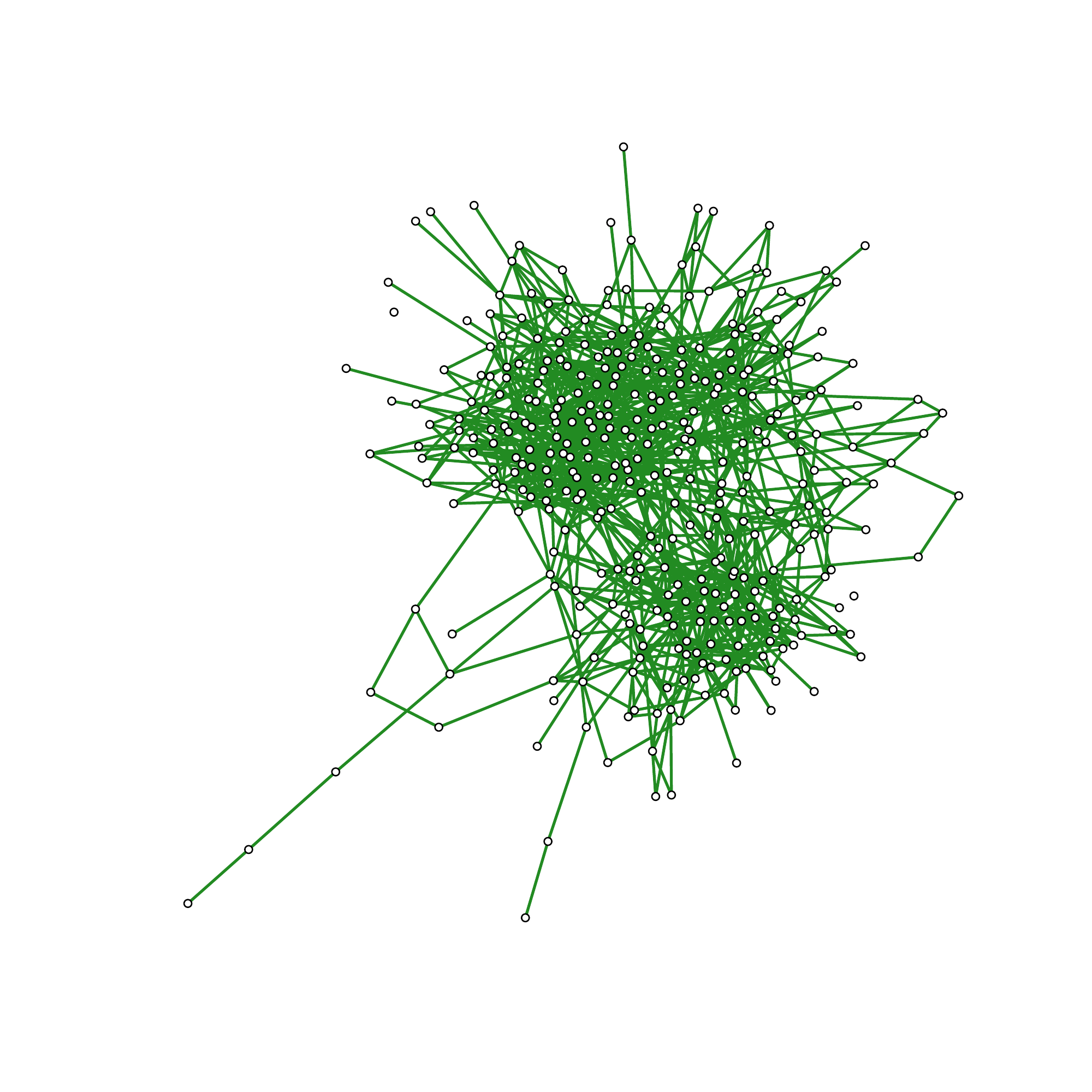}
    \vspace*{-2em}
    \subcaption{Village \# 11: Positive}
    \label{subfig:p11}
  \end{subfigure}
  \begin{subfigure}[b]{.49\textwidth}
    \includegraphics[width=.9\textwidth]{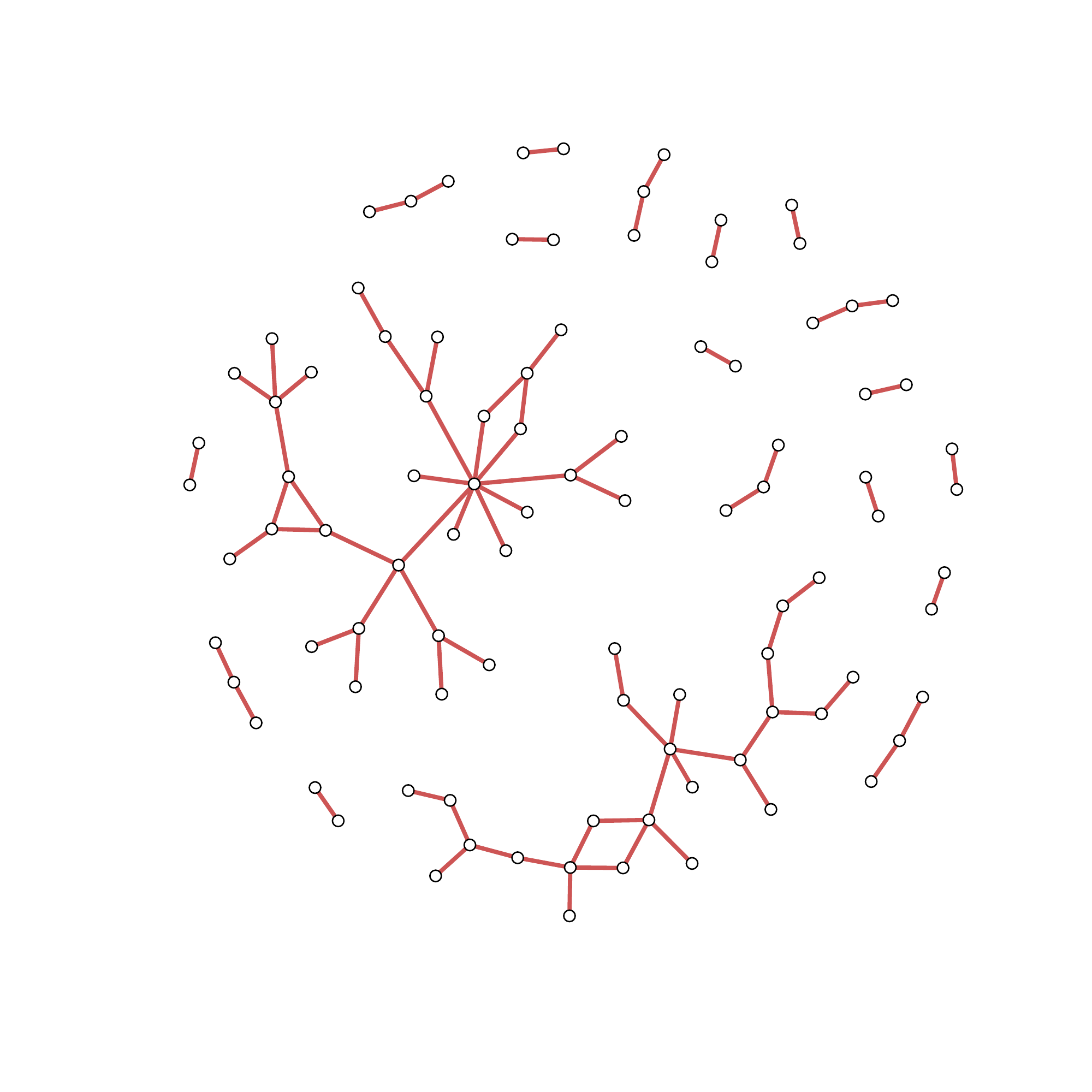}
    \vspace*{-2em}
    \subcaption{Village \# 26: Negative}
    \label{subfig:n26}
  \end{subfigure}
  \begin{subfigure}[b]{.49\textwidth}
    \includegraphics[width=.9\textwidth]{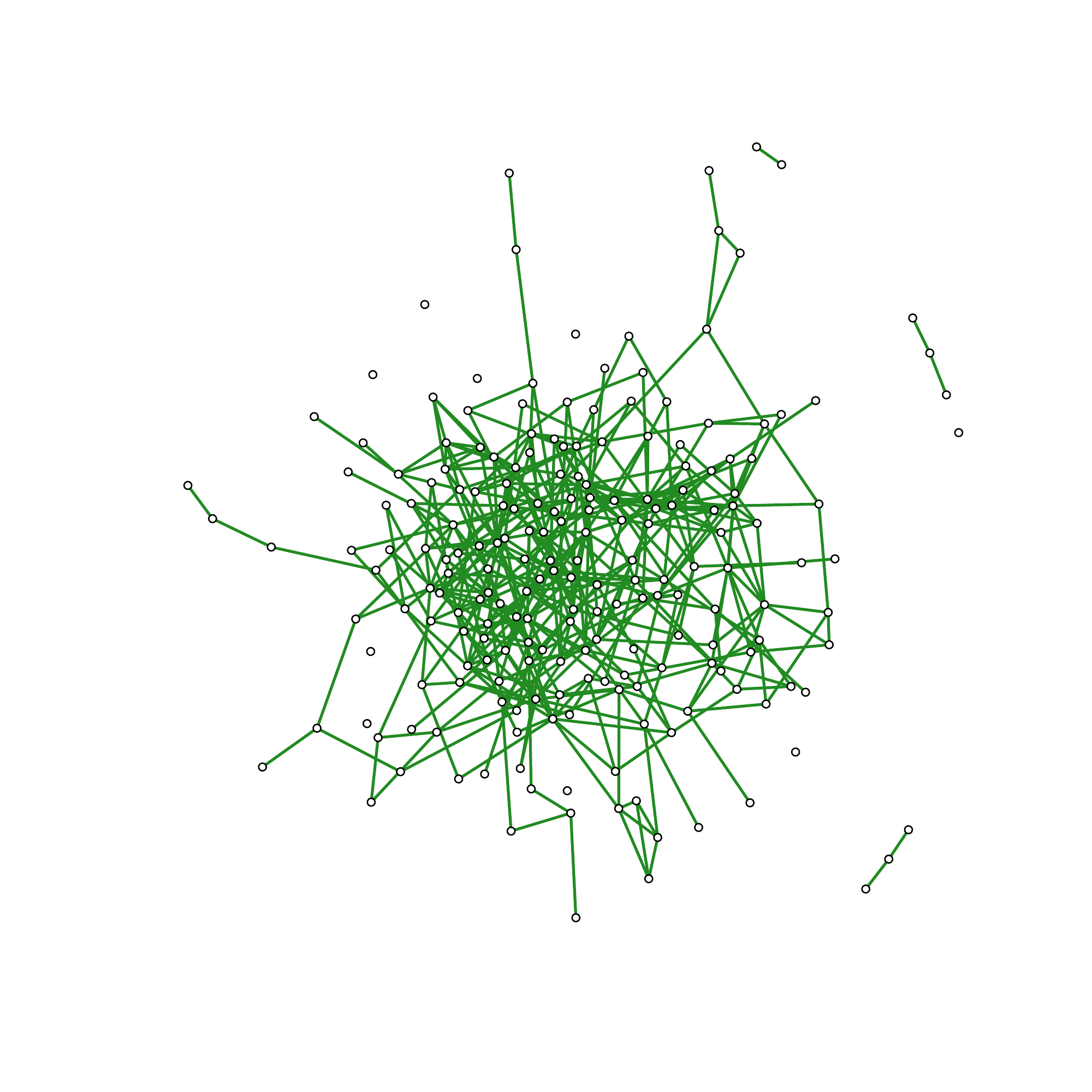}
    \vspace*{-2em}
    \subcaption{Village \# 26: Positive}
    \label{subfig:p26}
  \end{subfigure}
  \caption{
  Representitive examples of negative and positive subgraphs. In \subref{subfig:p11}, there is evidence of clustering into two communities. On the other hand, the negative subgraphs contain many components of very long trees. The largest component in \subref{subfig:n26}, for instance, would never be mistaken for a positive social network.
  }
  \label{fig:villages}
\end{figure}

More concretely, we calculated various graph statistics aggregated over the 32 villages, which we show in the first two rows of \Cref{tbl:summary_statistics}.
We see that negative ties are much less common than positive ties, exhibit very little transitivity, do not form long connected paths, nor do they coagulate into a giant component -- corroborating our observational conclusions.
In other words, they lack the common features found in positive social networks.

\begin{table}[h]
  \centering
  \footnotesize
  \begin{threeparttable}
  \begin{tabular}{$l^r^r^r^r^r^r}
    \toprule
    Type & \# Edges & Graph Density & Transitivity & Mean Path Length & \# Components \\ 
    \midrule
    Pos. & 331.9 & 0.056 & 0.26 & 3.59 & 1.66 \\ 
    Neg. & 19.7 & 0.002 & 0.01 & 1.87 & 8.44 \\ 
    \rowstyle{\itshape} Rnd. & - & - & 0.04 & 1.32 & 13.59 \\ 
    \bottomrule
  \end{tabular}
  \begin{tablenotes}
  \footnotesize
  \item[] \emph{Graph Density} is the ratio of edges to the number of possible edges
  \item[] \emph{\# Components} is the number of non-denegerate components, ignoring isolated vertices
  \end{tablenotes}
  \caption{Average values of select graph statistics across the two subgraphs (positive and negative) of the 32 village networks, as well as an \ER graph with the same number of edges and vertices as the negative subgraph.
  }
  \label{tbl:summary_statistics}
  \end{threeparttable}
\end{table}



To facilitate the more appropriate comparison with sparse random graphs, we generated, for each negative subgraph, an associated \ER graph under the $G(n,m)$ model with the same $n,m$ as the village graph (where $n$ is the number of nodes, and $m$ is the number of negative ties). The same graph statistics were calculated on this random graph sequence (shown in the third row (\emph{Rnd.}) of \Cref{tbl:summary_statistics}).
As expected, these numbers are very similar to that of the negative subgraphs.
Of particular note is that the transitivity in the random graphs is higher than that from the negative subgraph, suggesting that negative ties might be repelled from forming triangles.

\subsection{Interdependence between Negative and Positive Ties}

So far, we have considered negative ties in isolation, as an entity separate and independent of their positive counterparts.
However, such independent analysis can only provide a partial picture, as it ignores the positive landscape in which the negative ties are embedded. The same can be said for positive ties.
It is therefore crucial that we understand the interactions between positive and negative ties, as to start building the foundations for a simultaneous model of positive and negative relationships on a network.



Before we investigate balance, let us first test for structural differences between positive and negative ties in the networks. To that end, we shall test for a difference in the average embeddedness levels between the two ties, under the original null model, which we will refer to as the \emph{structural test}. 
We find that 17 graphs are significant at the 5\% level, out of a possible 28 graphs (4 graphs have no negative ties).
This reinforces our claim that negative and positive ties are structurally different.
More importantly, all 17 graphs found significant in structural differences were also significant for the old test.
If we look at a contingency table between these two tests (\Cref{subtbl:structvstan}), it is clear from the concentration of mass on the diagonal that the old test is essentially a test for structural differences, in disguise.
Thus, if we simply used the old test on this data, the conclusion, erroneously drawn, would be that balance is highly significant.



\begin{table}[h]
  \small
  \centering
  \begin{subtable}{.49\linewidth}
    \centering
    \begin{tabular}{r|cccc}
      \diagbox{Struct.}{Old} & Not Significant & Significant \\
      \midrule
      Not Significant & 9 & 2 \\
      Significant & 0 & 17 \\
    \end{tabular}
    \caption{Structural vs. Old Test}
    \label{subtbl:structvstan}
  \end{subtable}
  \begin{subtable}{.49\linewidth}
    \centering
    \begin{tabular}{r|cccc}
      \diagbox{New}{Old} & Not Significant & Significant \\
      \midrule
      Not Significant & 9 & 9 \\
      Significant & 0 & 10 \\
    \end{tabular}
    \caption{New Test vs. Old Test}
    \label{subtbl:stratvsstan}
  \end{subtable}
  \caption{Contigency table between the various tests.}
  \label{tbl:cont_table_sig}
\end{table}

Applying our new test to the dataset, we find that at the 5\% level, instead of the 19 graphs found significant from the old test, only half of them are actually significant.
The exact significance levels are shown in \Cref{tbl:individual_sig}.
While there is sufficient evidence to reject the null hypothesis that these social networks do not follow balance, it is clear that balance is not a ubiquitous force. In fact, with only a minority of villages exhibiting balance, the natural question then arises of what makes some social networks more receptive to balancing forces than others.

\begin{table}
  \centering
  \begin{threeparttable}
  \small
  \begin{tabular}{$l^l^l^l^l^l^l}
    \toprule
    ID & \multicolumn{2}{l}{Old} & \multicolumn{2}{l}{Structural} & \multicolumn{2}{l}{New} \\ 
    \midrule
  1 & 0.15 & \scriptsize{ } & 0.14 & \scriptsize{ } & 1.00 & \scriptsize{ } \\ 
    2 & 0.22 & \scriptsize{ } & 0.37 & \scriptsize{ } & 0.31 & \scriptsize{ } \\ 
    4 & 0.19 & \scriptsize{ } & 0.19 & \scriptsize{ } & 1.00 & \scriptsize{ } \\ 
    5 & 0.00 & \scriptsize{$\ast\ast$} & 0.02 & \scriptsize{$\ast$} & 0.01 & \scriptsize{$\ast$} \\ 
    \rowstyle{\bfseries} 6 & 0.00 & \scriptsize{$\ast{\ast}\ast$} & 0.00 & \scriptsize{$\ast{\ast}\ast$} & 0.07 & \scriptsize{.} \\ 
    7 & 0.24 & \scriptsize{ } & 0.22 & \scriptsize{ } & 1.00 & \scriptsize{ } \\ 
    \rowstyle{\bfseries} 8 & 0.01 & \scriptsize{$\ast$} & 0.01 & \scriptsize{$\ast$} & 1.00 & \scriptsize{ } \\ 
    9 & 0.00 & \scriptsize{$\ast{\ast}\ast$} & 0.06 & \scriptsize{.} & 0.00 & \scriptsize{$\ast{\ast}\ast$} \\ 
    10 & 0.00 & \scriptsize{$\ast{\ast}\ast$} & 0.01 & \scriptsize{$\ast\ast$} & 0.00 & \scriptsize{$\ast{\ast}\ast$} \\ 
    \rowstyle{\bfseries} 11 & 0.00 & \scriptsize{$\ast{\ast}\ast$} & 0.00 & \scriptsize{$\ast{\ast}\ast$} & 0.36 & \scriptsize{ } \\ 
    12 & 0.00 & \scriptsize{$\ast{\ast}\ast$} & 0.01 & \scriptsize{$\ast\ast$} & 0.02 & \scriptsize{$\ast$} \\ 
    13 & 0.00 & \scriptsize{$\ast{\ast}\ast$} & 0.00 & \scriptsize{$\ast{\ast}\ast$} & 0.00 & \scriptsize{$\ast\ast$} \\ 
    14 & 0.10 & \scriptsize{ } & 0.25 & \scriptsize{ } & 0.14 & \scriptsize{ } \\ 
    15 & 0.24 & \scriptsize{ } & 0.50 & \scriptsize{ } & 0.22 & \scriptsize{ } \\ 
    \rowstyle{\bfseries} 16 & 0.04 & \scriptsize{$\ast$} & 0.04 & \scriptsize{$\ast$} & 1.00 & \scriptsize{ } \\ 
    19 & 0.07 & \scriptsize{.} & 0.38 & \scriptsize{ } & 0.06 & \scriptsize{.} \\ 
    20 & 0.52 & \scriptsize{ } & 0.51 & \scriptsize{ } & 1.00 & \scriptsize{ } \\ 
    21 & 0.00 & \scriptsize{$\ast{\ast}\ast$} & 0.03 & \scriptsize{$\ast$} & 0.00 & \scriptsize{$\ast{\ast}\ast$} \\ 
    \rowstyle{\bfseries} 22 & 0.00 & \scriptsize{$\ast{\ast}\ast$} & 0.00 & \scriptsize{$\ast{\ast}\ast$} & 0.39 & \scriptsize{ } \\ 
    23 & 0.00 & \scriptsize{$\ast{\ast}\ast$} & 0.00 & \scriptsize{$\ast{\ast}\ast$} & 0.01 & \scriptsize{$\ast$} \\ 
    \rowstyle{\bfseries} 24 & 0.00 & \scriptsize{$\ast{\ast}\ast$} & 0.00 & \scriptsize{$\ast{\ast}\ast$} & 0.06 & \scriptsize{.} \\ 
    \rowstyle{\bfseries} 25 & 0.01 & \scriptsize{$\ast\ast$} & 0.01 & \scriptsize{$\ast\ast$} & 1.00 & \scriptsize{ } \\ 
    26 & 0.00 & \scriptsize{$\ast{\ast}\ast$} & 0.00 & \scriptsize{$\ast{\ast}\ast$} & 0.00 & \scriptsize{$\ast\ast$} \\ 
    \rowstyle{\bfseries} 27 & 0.00 & \scriptsize{$\ast\ast$} & 0.00 & \scriptsize{$\ast\ast$} & 1.00 & \scriptsize{ } \\ 
    28 & 0.29 & \scriptsize{ } & 0.27 & \scriptsize{ } & 1.00 & \scriptsize{ } \\ 
    29 & 0.00 & \scriptsize{$\ast{\ast}\ast$} & 0.06 & \scriptsize{.} & 0.00 & \scriptsize{$\ast\ast$} \\ 
    30 & 0.00 & \scriptsize{$\ast{\ast}\ast$} & 0.01 & \scriptsize{$\ast\ast$} & 0.00 & \scriptsize{$\ast{\ast}\ast$} \\ 
    \rowstyle{\bfseries} 31 & 0.01 & \scriptsize{$\ast\ast$} & 0.04 & \scriptsize{$\ast$} & 0.07 & \scriptsize{.} \\ 
     \bottomrule
  \end{tabular}
  \begin{tablenotes}
  \footnotesize
  \item[] .~---~$0.05 \leq p < 0.1$, *~---~$0.01 \leq p < 0.05$, **~---~$0.005 \leq p < 0.01$,
  \item[] ***~---~$0 \leq p < 0.005$
  \end{tablenotes}
  \caption{%
  Table showing the results of three statistical tests on each village network.
  The nine rows in bold are those where the old test is significant while the new one is not.}
  \label{tbl:individual_sig}
  \end{threeparttable}
\end{table}


%% file: Chapters/conclusion.tex

\section{Conclusion and Future Work}



Models of social network structure and function would benefit from not ignoring negative ties. 
Our hope is that, \hl{in much the same way that moving from $\mathbb{R}$ to $\mathbb{C}$ produces new insights and simplifications}, the extension to signed graphs can provide simple mechanisms to hitherto complicated models.

Balance theory is potentially one such mechanism.
In this paper, we introduced a new test for balance that, unlike the standard test currently being used, takes into consideration the different behaviors of negative and positive ties.
We showed through theoretical analysis and simulations that our test outperforms the original.
We applied our test to a novel dataset of village social networks in rural Honduras, and found that balance theory holds true in a minority of the villages, though the villages varied in their extent of balance.






There is much more to be done in this nascent field of understanding and modeling signed social networks.
In addition to the task of understanding the causes of the distribution of balancing effects across social networks, we also have the following future directions:
\begin{itemize}
  \item 
  The count of unbalanced triangles is by its very nature a global measure of balance. For reasonable sized graphs such as our village dataset, a global measure is appropriate.
  However, once we move towards much larger social networks -- especially those on the order of, say, Facebook's social graph -- the assumption that there is one measure of balance across the entire graph is no longer tenable. This regime requires a completely new framework for measuring and interpreting balance, the first step of which is to define a new local definition of balance.
  \item The definition of balance as a function of the product of signs extends effortlessly to higher order cycles. 
  Here, we have restricted our attention to triads, as we think the first order is the most important (and most plausible).
  There has been some work considering higher order cycles \citep{Estrada:2014vb,iosifidis2018}, but little justification for doing so.
  Similarly, while balance theory is the most natural means of relating negative and positive ties, perhaps there are other types of relations possible.
  \item Stratification is only one way to account for the differences in negative and positive ties when performing a statistical test.
  An alternative method would be to incorporate existing network models of tie formation, such as exponential random graph models.
  We preferred to adopt a nonparametric approach here to minimize the number of model assumptions, but we leave this potential extension to future work.
  \item In this manuscript, we have considered \emph{undirected} signed graphs.
  One could extend our results to directed graphs, but this leads to a combinatorial explosion in the number of different triangle arrangements, and so one loses the simple interpretation of the 4 different states.
  Another extension is to consider weighted signed edges, not just binary signed edges. There, the question is how best to generalize the notion of balance in this setting -- one possibility is to take the product of the edges as a measure of the balance of that triangle.
\end{itemize}

%% file: Chapters/appendix.tex

\section{Proof of \texorpdfstring{\Cref{prop:clt_radamacher}}{Proposition \ref{prop:clt_radamacher}}}
\label{sec:appendix_proof}

Our goal in this section is to prove that, under \Cref{modl:new}, the joint distribution of $(T^{(n)}, M^{(n)})$ is asymptotically normally distributed. For simplicity, we drop the index $n$ most of the time.

Let us first sketch the main ideas of the proof. We begin by decomposing the centered versions of the $T_{\alpha}$'s in the spirit of a Hoeffding decomposition (\Cref{lem:decomposition}) into terms depending on one, two or three different $X_i$. This decomposition enables us to place our problem into the framework of functionals of Rademacher random variables. In \Cref{prop:zheng_clt} we show that the fourth moments of the terms from the decomposition converge, along with some additional upper bounds. This turns out to be enough to conclude that the terms live in a fixed \emph{Rademacher chaos} and so we obtain asymptotic normality via the Malliavin-Stein approach, in the form of \citep{Zheng:2017tr}. A simple recomposition of the decomposed random variables completes the proof of \Cref{prop:clt_radamacher}.


Let us first establish some additional notation. Recall that we are working under \Cref{modl:new}, so that the edge signs are independent and given by $X_i \sim \operatorname{Rad}\left( 1-p_{\varepsilon_i} \right)$. Then,
\begin{align}
  r_i \coloneqq \E{ X_i } = 1 - 2 p_{\varepsilon_i}, \qquad s_i^2 \coloneqq \V{ X_i } = 4 p_{\varepsilon_i} (1 - p_{\varepsilon_{i}}),
\end{align}
and we can define the normalized $X_i$ by
$\widetilde{X}_i \coloneqq \frac{X_i - r_i}{s_i}.$
On multiple occasions in the forthcoming proofs, it will be more intuitive to work with the Bernoulli random variable signifying the presence of a negative edge at position $i$, $Y_i \coloneqq \frac{1 - X_i}{2} \sim \operatorname{Bern}(p_{\varepsilon_i})$, rather than the Rademacher random variable $X_i$.
Finally, define $\triangle_{ijk} \coloneqq \I{(i,j,k) \in \triangle}$. Then the number of triangles with 1 negative tie is given by
\begin{align}
  T_1
  &= \frac{1}{2!} \sum_{(i,j,k) \in \triangle} Y_i (1-Y_j) (1- Y_k)
  \\&= \frac{1}{2}\sum_{(i,j,k) \in E^3} \triangle_{ijk} Y_i (1-Y_j) (1-Y_k), \label{eqn:defn_t1}
\end{align}
where the factor $\frac{1}{2}$ comes from overcounting.
Similarly,
\begin{align}
  T_2 &= \frac{1}{2}\sum_{(i,j,k) \in E^3} Y_i Y_j (1-Y_k),\,\,\,\, T_3 = \frac{1}{6} \sum_{(i,j,k) \in E^3} Y_i Y_j Y_k.
\end{align}

\begin{lem}\label{lem:decomposition}
  We have for $\alpha,\beta=1,2,3$ the decompositions $T_\alpha - \E{T_\alpha} = \sum_{\beta=1}^{3} T_{\alpha,\beta}$ with 
  \begin{align}
    T_{\alpha,1} &\coloneqq \sum_{i \in E} t_{\alpha,1}(i) \widetilde{X}_{i}, \\
    T_{\alpha,2} &\coloneqq \sum_{(i,j) \in E^2} t_{\alpha,2}(i,j) \widetilde{X}_{i} \widetilde{X}_{j},\\
    T_{\alpha,3} &\coloneqq \sum_{(i,j,k) \in E^{3}} t_{\alpha,3}(i,j,k) \widetilde{X}_{i} \widetilde{X}_{j} \widetilde{X}_{k},
  \end{align}
  and
  \begin{align}
    \begin{aligned}
    t_{1,1}(i) &\coloneqq s_{i} \sum_{(j,k) \in E^2} \frac{\triangle_{ijk}}{16} (1-3p_{\varepsilon_j})(1-p_{\varepsilon_k}),
    & t_{1,2}(i,j) &\coloneqq s_{i} s_{j} \sum_{k \in E}\frac{\triangle_{ijk}}{16} (3 p_{\varepsilon_k} - 2), \\
    t_{2,1}(i) &\coloneqq s_{i} \sum_{(j,k) \in E^2} \frac{\triangle_{ijk}}{16} p_{\varepsilon_j} (3p_{\varepsilon_k} - 2),
    & t_{2,2}(i,j) &\coloneqq s_{i} s_{j} \sum_{k \in E}\frac{\triangle_{ijk}}{16} (1-3p_{\varepsilon_k}), \\
    t_{3,1}(i) &\coloneqq s_{i} \sum_{(j,k) \in E^2} \frac{\triangle_{ijk}}{16} \left( -p_{\varepsilon_j} p_{\varepsilon_k} \right), &
    t_{3,2}(i,j) &\coloneqq s_{i} s_{j} \sum_{k \in E}\frac{\triangle_{ijk}}{16} p_{\varepsilon_k}, \\
    \end{aligned}
  \end{align}
  \begin{align}
    t_{1,3}(i,j,k) = -t_{2,3}(i,j,k) = \frac{1}{3} t_{3,3}(i,j,k) &\coloneqq - s_i s_j s_k \frac{\triangle_{ijk}}{16}.
  \end{align}
\end{lem}

\begin{proof}
  For clarity of the proof, it will be helpful to introduce the temporary variables $Y_i' = 1 - Y_i$, so that $\E{ Y_i'} = p_{\varepsilon_i}' \coloneqq (1 - p_{\varepsilon_i})$.
  The normalized versions of $Y_i, Y_i', X_i$ satisfy $\widetilde{Y}_{i} = -\frac{1}{2} \widetilde{X}_{i}$, $   \widetilde{Y}_{i}' = \frac{1}{2} \widetilde{X}_{i}$.
 By independence $\E{T_1}= \sum_{(i,j,k) \in E^{3}} \frac{\triangle_{ijk}}{2} p_{\varepsilon_i} p_{\varepsilon_j}' p_{\varepsilon_k}'$ and therefore
  \begin{align}
    &T_1 - \E{ T_1 }
    = \sum_{(i,j,k) \in E^{3}} \frac{\triangle_{ijk}}{2} Y_i Y_j' Y_k' - \sum_{(i,j,k) \in E^{3}} \frac{\triangle_{ijk}}{2} p_{\varepsilon_i} p_{\varepsilon_j}' p_{\varepsilon_k}'
    \\&=
    \sum_{(i,j,k) \in E^{3}} \frac{\triangle_{ijk}}{2} \left[ (Y_i - p_{\varepsilon_i}) (Y_{j}' - p_{\varepsilon_j}') (Y_{k}' - p_{\varepsilon_{}k}') \right.
    \\&\quad + p_{\varepsilon_i} (Y_j' - p_{\varepsilon_j}') (Y_k' - p_{\varepsilon_k}')
     + (Y_i - p_{\varepsilon_i}) p_{\varepsilon_j}' (Y_k' - p_{\varepsilon_k}')
     + (Y_i - p_{\varepsilon_i}) (Y_j' - p_{\varepsilon_j}') p_{\varepsilon_k}'
    \\&\quad - p_{\varepsilon_i} p_{\varepsilon_j}' (Y_k' - p_{\varepsilon_k}')
     - p_{\varepsilon_i} (Y_j' - p_{\varepsilon_j}') p_{\varepsilon_k}'
     - (Y_i - p_{\varepsilon_i}) p_{\varepsilon_j}' p_{\varepsilon_k}' \left. \vphantom{ \frac{1}{2} } \right].
 \end{align}
 Rewriting this in terms of the normalized random variables $\widetilde{X}_i$ shows that this is equal to
 \begin{align}
 	& \sum_{(i,j,k) \in E^{3}} \frac{\triangle_{ijk}}{16} \left[ s_{i} s_{j} s_{k} \cdot (-\widetilde{X}_{i}) \widetilde{X}_{j} \widetilde{X}_k \right.
     \\&\quad + s_{j} s_{k} \cdot p_{\varepsilon_i} \widetilde{X}_j \widetilde{X}_k
      + s_{i} s_{k} \cdot (- \widetilde{X}_{i}) p_{\varepsilon_j}' \widetilde{X}_{k}
      + s_{i} s_{j} \cdot (- \widetilde{X}_{i}) \widetilde{X}_{j} p_{\varepsilon_k}'
     \\&\quad - s_{k} \cdot p_{\varepsilon_i} p_{\varepsilon_j}' \widetilde{X}_{k}
      - s_{j} \cdot p_{\varepsilon_i} \widetilde{X}_{j} p_{\varepsilon_k}'
      - s_{i} \cdot (- \widetilde{X}_i) p_{\varepsilon_j}' p_{\varepsilon_k}' \left. \vphantom{ \frac{1}{16} } \right]
     \\&= 
     \sum_{(i,j,k) \in E^{3}} \frac{\triangle_{ijk}}{16} \big[ - s_{i} s_{j} s_{k} \cdot \widetilde{X}_{i} \widetilde{X}_{j} \widetilde{X}_k
      + s_{i} s_{j} \cdot (p_{\varepsilon_k} - 2 p_{\varepsilon_k}') \widetilde{X}_i \widetilde{X}_j
      \\&\quad - s_{i} \cdot (2p_{\varepsilon_j} p_{\varepsilon_k}' - p_{\varepsilon_j}' p_{\varepsilon_k}' ) \widetilde{X}_{i}
      \big] = T_{1,1} + T_{1,2} + T_{1,3}.
\end{align}

  With similar calculations for $T_2$ and $T_3$, we get the claimed expressions.
\end{proof}
By independence of the $\widetilde{X}_i$, this decomposition provides a simple form for the variances of the $T_{\alpha, \beta}$. For instance, 
\begin{align}
  \V{ T_{\alpha,3}} &= \sum_{(i,j,k) \in E^3} t^2_{\alpha,3}(i,j,k).
\end{align}
Recall that $M_l$ for $l=0,\dots,L_{-}$ denotes the number of negative ties in embeddedness level $l$. This means $M_l = \sum_{i \in E} \I{ i \in E_{l} } Y_i$, where $E_{l}$ is the set of edges with embeddedness $l$. Define the centered versions by
\begin{align}
  \check{M}_{l} \coloneqq M_{l} - \E{ M_{l} }
  &= \sum_{ i \in E } \I{ i \in E_{l} } \left( Y_i - p_{\varepsilon_i} \right)
  = \sum_{ i \in E } h_{l}(i) \widetilde{X}_{i},
\end{align}
where $h_{l}(i) \coloneqq -\frac{s_i}{2} \I{ i \in E_{l} }$, with variances
\begin{align}
  \V{ \check{M}_{l} } 
  &= \sum_{i \in E} \I{ i \in E_{l}} \frac{s_i^2}{4}
  \\&
  = n_l p_l (1 - p_l) = m_l (1 - p_l). \label{eqn:var_m}
\end{align}
\begin{prop}\label{prop:zheng_clt}
  Consider the vector
  \begin{align}
    B^{(n)} \coloneqq
    \frac{1}{\sqrt{n}}
    \left( T^{(n)}_{1,1}, T^{(n)}_{1,2}, T^{(n)}_{1,3}, T^{(n)}_{2,1}, T^{(n)}_{2,2}, T^{(n)}_{2,3}, T^{(n)}_{3,1}, T^{(n)}_{3,2}, T^{(n)}_{3,3},
    \check{M}^{(n)}_{0}, \ldots, \check{M}^{(n)}_{L_{-}} \right)^{\top}.
  \end{align}
  with covariance matrix $\Gamma^{(n)} \in \mathbb{R}^{(10+L_{-})\times(10+L_{-})}$.
  Under \Cref{ass:all_embed_limit} we have the following:
  \begin{enumerate}
    \item $\Gamma^{(n)}$ converges to a matrix $\Gamma$ component-wise as $n \to \infty$;
    \item For $\alpha = 1,2,3$ and $l = 0, \ldots, L_{-}$, we have
    \begin{align}
      \sup_{i \in E} \frac{1}{n} h_{l}^2 (i) &\to 0, \,\,\,\, \sup_{i \in E} \frac{1}{n} t_{\alpha,1}^2(i) \to 0, \label{eqn:sup_condition_1} \\
      \sup_{i \in E} \frac{1}{n} \sum_{j \in E} t_{\alpha,2}^2(i,j) &\to 0, \,\,\,\,\sup_{i \in E} \frac{1}{n} \sum_{(j,k) \in E^2} t_{\alpha,3}^2(i,j,k) \to 0 \label{eqn:sup_condition_2}.
    \end{align}
    \item For $\alpha,\beta = 1,2,3$ and $l = 0, \ldots, L_{-}$ we have
    \begin{align}
      \frac{1}{n^2} \E{ (T^{(n)}_{\alpha,\beta})^4 } &\to 3 \Gamma^2_{3(\alpha-1)+\beta, 3(\alpha-1)+\beta}, \\
      \frac{1}{n^2} \E{ (\check{M}^{(n)}_{l})^4 } &\to 3 \Gamma^2_{10+l, 10+l}.
    \end{align}
  \end{enumerate}
\end{prop}
\begin{proof} For simplicity, in the following we drop the index $n$.
\begin{enumerate}[wide, labelwidth=!, labelindent=0pt]
    \item 
    We have for the covariances with $\alpha, \beta = 1,2,3$ and $l = 0, \ldots, L_{-}$
    \begin{align}
      \frac{1}{n} \C{ T_{\alpha,1}, T_{\beta,1} }
      &= \frac{1}{n} \sum_{i \in E} t_{\alpha,1}(i) t_{\beta,1}(i), \label{eqn:cov_t1}
      \\
      \frac{1}{n} \C{ T_{\alpha,2}, T_{\beta,2}}
      &= \frac{1}{n} \sum_{(i,j) \in E^2} t_{\alpha,2}(i,j) t_{\beta,2}(i,j),  \label{eqn:cov_t2}
      \\
      \frac{1}{n} \C{ T_{\alpha,3}, T_{\beta,3} }
      &= \frac{1}{n} \sum_{(i,j,k) \in E^3} t_{\alpha,3}(i,j,k) t_{\beta,3}(i,j,k), \label{eqn:cov_t3}
      \\
      \frac{1}{n} \C{ T_{\alpha,1}, \check{M}_{l}}
      &= \frac{1}{n} \sum_{i \in E} t_{\alpha,1}(i) h_{l}(i), \label{eqn:cov_t_m} \\
      \frac{1}{n} \C{ \check{M}_{l}, \check{M}_{l}}
      &= \frac{1}{n} \sum_{i \in E} h_{l}^2(i), \label{eqn:cov_m}
    \end{align}
    while all other covariances are zero. In order to see that these terms converge, it suffices to show that they are functions of (finite) linear combinations of the $p_l$ and of 
    \begin{align}
      q_{l_1,l_2,l_3}=\frac{1}{n} &\sum_{(i,j,k) \in E_{l_1,l_2,l_3}} \triangle_{ijk},
      \,\,\,\,
      u_{l_1,l_2,l_3,l_4,l_5}=
      \frac{1}{n}\sum_{i \in E_{l_1}}
      \left(\sum_{(j,k) \in E_{l_2,l_3}} 
      \triangle_{ijk}\right)
      \left(\sum_{(j',k') \in E_{l_4,l_5}}
      \triangle_{ij'k'}\right),
    \end{align}
    for $l_1,l_2,l_3,l_4,l_5=0,\dots,L_{-}$, which we assume to converge by \Cref{ass:all_embed_limit}.
    With respect to \cref{eqn:cov_t1} observe that the $t_{\alpha,1}(i)$ are of the form
    \begin{align}
      s_i \sum_{(j,k) \in E^2} \triangle_{ijk} f_{\alpha}(p_{\varepsilon_j}, p_{\varepsilon_k})
      &= s_i \sum_{\substack{ l_{2}, l_{3} = 0 }}^{L_{-}} f_{\alpha}(p_{l_{2}}, p_{l_{3}}) \sum_{(j,k) \in E_{l_2,l_3}}\triangle_{ijk}
    \end{align}
   for some functions $f_\alpha$. Hence, the covariances satisfy
    \begin{align}
      \frac{1}{n} \C{ T_{\alpha, 1}, T_{\beta, 1} }
      &= \sum_{i \in E} s_i^2 \left( \sum_{(j,k) \in E^2} \triangle_{ijk} f_{\alpha}(p_{\varepsilon_j}, p_{\varepsilon_k}) \right) \left( \sum_{(j,k) \in E^2} \triangle_{ijk} f_{\beta}(p_{\varepsilon_j}, p_{\varepsilon_k}) \right)
      \\&= \sum_{l_1=0}^{L_{-}} 4 p_{l_1} (1 - p_{l_1})
      F(l_1),
    \end{align}
    where
    \begin{align}
      F(l_1) &= 
      \frac{1}{n}
      \sum_{i \in E_{l_{1}}}
      \left( \sum_{l_2, l_3 = 0}^{L_{-}} f_{\alpha}(p_{l_2}, p_{l_3}) \sum_{(j,k) \in E_{l_2,l_3}} \triangle_{ijk}  \right)
      \left( \sum_{l_2, l_3 = 0}^{L_{-}} f_{\beta}(p_{l_2}, p_{l_3}) \sum_{(j,k) \in E_{l_2,l_3}} \triangle_{ijk}  \right)
      \\&= 
      \frac{1}{n}
      \sum_{i \in E_{l_{1}}}
      \sum_{l_2, l_3 = 0}^{L_{-}}
      \sum_{l_4, l_5 = 0}^{L_{-}} f_{\alpha}(p_{l_2}, p_{l_{3}}) f_{\beta}(p_{l_{4}}, p_{l_{5}})
      \sum_{(j,k) \in E_{l_2,l_3}} 
      \triangle_{ijk}
      \sum_{(j',k') \in E_{l_4,l_5}}
      \triangle_{ij'k'}
      \\&=\sum_{l_2, l_3 = 0}^{L_{-}}
      \sum_{l_4, l_5 = 0}^{L_{-}} f_{\alpha}(p_{l_2}, p_{l_{3}}) f_{\beta}(p_{l_{4}}, p_{l_{5}})
      u_{l_1, l_2, l_3, l_4, l_5}.
    \end{align}
    Next, with respect to \cref{eqn:cov_t2},
    \begin{align}  
      & \frac{1}{n} \C{ T_{\alpha, 2}, T_{\beta, 2} } = \frac{1}{n} \sum_{(i,j) \in E^2} s_{i}^2 s_{j}^2
      \left( \sum_{k \in E} \triangle_{ijk} f_{\alpha}(p_{\varepsilon_k}) \right)
      \left( \sum_{k' \in E} \triangle_{ijk'} f_{\beta}(p_{\varepsilon_k'}) \right)
      \\&=
      \frac{1}{n} \sum_{l_1, l_2 = 0}^{L_{-}} 16 p_{l_1} (1-p_{l_1}) p_{l_2} (1-p_{l_2})
      \sum_{\substack{ i \in E_{l_1}, \\ j \in E_{l_2} }}
      \sum_{\substack{k \in E, \\ k' \in E}}
      \triangle_{ijk} \triangle_{ijk'} f_{\alpha}(p_{\varepsilon_k}) f_{\beta}(p_{\varepsilon_{k'}}).
      \intertext{For fixed $i,j$ there is only one $k$ with $\triangle_{ijk} = 1$. Thus, the last line equals}
      &
      \frac{1}{n} \sum_{l_1, l_2 = 0}^{L_{-}} 16 p_{l_1} (1-p_{l_1}) p_{l_2} (1-p_{l_2})
      \sum_{(i,j) \in E_{l_1,l_2}}
      \sum_{k \in E}
      \triangle_{ijk} f_{\alpha}(p_{\varepsilon_k}) f_{\beta}(p_{\varepsilon_{k}})
      \\&= 
      \sum_{l_1, l_2, l_3 = 0}^{L_{-}} 16 p_{l_1} (1-p_{l_1}) p_{l_2} (1-p_{l_2})
      f_{\alpha}(p_{l_{3}}) f_{\beta}(p_{l_{3}}) q_{l_1,l_2,l_3},
     \end{align}
    Similarly, for some $C_{\alpha,\beta}$, \cref{eqn:cov_t3} can be written as
    \begin{align}
      & \frac{1}{n} \C{ T_{\alpha, 3}, T_{\beta, 3} }= \frac{1}{n} C_{\alpha, \beta} \sum_{(i,j,k) \in E^{3}} s_{i}^2 s_{j}^2 s_{k}^2 \triangle_{ijk}
      \\&= C_{\alpha, \beta} \sum_{l_1, l_2, l_3 = 0}^{L_{-}} 4^3 p_{l_1}(1 - p_{l_1}) p_{l_2}(1 - p_{l_2}) p_{l_3}(1 - p_{l_3}) q_{l_1,l_2,l_3}.
    \end{align} 
    On the other hand, the covariances between $T_{\alpha,1}$ and $\check{M}_{l}$ in \cref{eqn:cov_t_m} satisfy
    \begin{align}
      \frac{1}{n} \C{ T_{\alpha,1}, \check{M}_{l} }
      &= - \frac{1}{n} \sum_{i \in E} \I{ i \in E_{l} } \frac{s_i^2}{2} \sum_{(j,k) \in E^2} \triangle_{ijk} f_{\alpha}( p_{\varepsilon_j}, p_{\varepsilon_k} )
      \\&= -\sum_{l_2, l_3 = 0}^{L_{-}} 2 p_{l} (1 - p_{l})
      f_{\alpha}(p_{l_2}, p_{l_{3}}) q_{l_1,l_2,l_3}.
    \end{align}
    Finally, with respect to  \cref{eqn:var_m}, $\frac{1}{n} \V{ \check{M}_{l} }= \frac{n_l}{n} p_l (1 - p_l)$     converges by convergence of $n_l/n$ (cf. \cref{eqn:n_l}) and of the $p_l$.
    \item Starting with \cref{eqn:sup_condition_1}, we have
    \begin{align}
      \sup_{i \in E} \frac{1}{n} h_l^2(i) 
      &= \frac{\I{ i \in E_{l}} \frac{s_{i}^2}{4}}{n} = \frac{p_{l} (1 - p_{l})}{n} \to 0,
    \end{align}
    since, by \Cref{ass:all_embed_limit}, the $p_{l}$ converge and are therefore bounded. On the other hand, observe that
    \begin{align}
      \varepsilon_i^2 = 
      \left( \frac{1}{2} \sum_{(j,k) \in E^2} \triangle_{ijk} \right)^2
      >
      \begin{cases}
        t^2_{\alpha, 1} (i) \\
        \sum_{j \in E} t^2_{\alpha, 2}(i,j) \\
        \sum_{(j,k) \in E^2} t^2_{\alpha, 3}(i,j,k).
      \end{cases}
    \end{align}
    Applying \Cref{ass:emb2} gives
    \begin{align}
      \sup_{i \in E} \frac{1}{n} t^2_{\alpha, 1}(i) 
      &\leq \frac{1}{n} \sup_{i \in E} \varepsilon_i^2 \to 0.
     \end{align}
  The two statements in \cref{eqn:sup_condition_2} follow similarly.
    \item Let us first consider $T_{1,1}$. We have
    \begin{align}
      \frac{1}{n^2} \E{ T_{1,1}^{4} }
      &= \frac{1}{n^2}\sum_{(i,j,k,l) \in E^{4}} t_{1,1}(i) t_{1,1}(j) t_{1,1}(k) t_{1,1}(l) \E{ \widetilde{X}_{i} \widetilde{X}_{j} \widetilde{X}_{k} \widetilde{X}_{l} }
      \\&= \frac{1}{n^2}\sum_{i \in E} t_{1,1}^{4}(i) \E{ \widetilde{X}_{i}^4 } + \frac{1}{n^2}\frac{\binom{4}{2}}{2} \sum_{(i,j) \in E^{2}, i \neq j} t_{1,1}^2(i) t_{1,1}^2(j) \E{ \widetilde{X}_{i}^2 } \E{ \widetilde{X}_{j}^2 }
      \\&= \frac{1}{n^2}\sum_{i \in E} t_{1,1}^{4}(i) \left[ \E{ \widetilde{X}_{i}^4 } - 3 \right] + 3 (\Gamma^{(n)}_{1,1})^2.
    \end{align}
    Focusing on the first term, note that 
    \begin{align}
      \frac{1}{n^2} \sum_{i \in E} t^{4}_{1,1}(i)
      &\leq \frac{1}{n} \sup_{i \in E} t^2_{1,1}(i) \frac{1}{n} \sum_{i \in E} t^2_{1,1}(i)= \frac{1}{n} \sup_{i \in E} t^2_{1,1}(i) \Gamma^{(n)}_{1,1}.
    \end{align}
    We already know, however, from part $(2.)$ of this proof that $\frac{1}{n} \sup_{i \in E} t^2_{1,1}(i) \to 0$. Thus, the first term converges to 0, which means that $\frac{1}{n^2}\E{ T_{1,1}^{4} } \to 3 ( \Gamma_{1,1})^2$, as required.
    The remaining limits follow similarly.
  \end{enumerate}
\end{proof}
We are now ready to prove \Cref{prop:clt_radamacher}.
\begin{proof}[Proof of \Cref{prop:clt_radamacher}] The three conditions in  \Cref{prop:zheng_clt} imply by Theorem 1.1 of \citet{Zheng:2017tr} that $B^{(n)} \xrightarrow{\mathcal{D}} \mathcal{N}(\mathbf{0}, \Gamma)$ as $n \to \infty$. Define the matrix $Q: \mathbb{R}^{10+L_{-}} \to \mathbb{R}^{4+L_{-}}$ by
  \begin{align}
    Q=\begin{pmatrix}
     1 & 1 & 1\\
     &  &  & 1 & 1 & 1\\
     &  &  &  &  &  & 1 & 1 & 1\\
     &  &  &  &  &  &  &  &  & I_{L_{-}+1}
    \end{pmatrix},
  \end{align}
  where $I_{L_{-}+1}$ is the identity matrix of size $L_{-}+1$. We conclude that
  \begin{align}
    \left(\widetilde{T}^{(n)}, \widetilde{M}^{(n)}\right)^\top = QB^{(n)} \xrightarrow{\mathcal{D}} \mathcal{N}(0, \Sigma),
  \end{align}
  with covariance matrix $ \Sigma = Q \Gamma Q^{\top}$. We have
  \begin{align}
    \Sigma = \lim_{n\rightarrow\infty} \Sigma^{(n)}, \label{eqn:sigma3}
  \end{align}
   where $\Sigma^{(n)}$ is the covariance matrix of $QB^{(n)}$. The form of $Q$ yields for $\alpha,\beta=1,2,3$, $l=0,\dots,L_{-}$
  \begin{align}
    \Sigma^{(n)}_{\alpha,\alpha} &= \frac{1}{n} \sum_{\beta = 1}^{3} \V{ T^{(n)}_{\alpha,\beta} },\,\,\,\,\Sigma^{(n)}_{\alpha,\beta} = \frac{1}{n} \sum_{\gamma = 1}^{3} \C{ T^{(n)}_{\alpha, \gamma}, T^{(n)}_{\beta, \gamma} },
    \\
    \Sigma^{(n)}_{\alpha, 4+l} &= \frac{1}{n} \C{ T^{(n)}_{\alpha,1}, \check{M}^{(n)}_{l} }, \,\,\,\, \Sigma^{(n)}_{4+l,4+l}= \V{\check{M}^{(n)}_l}, \label{eqn:sigma2}
  \end{align}
  and $\Sigma^{(n)}_{4+l,4+j}=0$ for $l,j=0,\dots,L_{-}$ with $l\neq j$. 
\end{proof}


\section{Remaining Proofs}
\label{sec:appendix_lemmas}

\begin{proof}[Proof of \Cref{lem:cond_dist}]
We have to show for $x=(x_i)_{i\in E}$, $x_i \in \{-1,1\}$, that 
\begin{align}
	\P{W_{\pi} = x} = \P{X=x\,\mid\, M=m}.
\end{align}
By independence of edges $X_i$ for different levels of embeddedness it is enough to show
\begin{align}
	\P{W_{\tau_l} = x} = \P{(X_i)_{i\in E_l}=x\,\mid\, M_l=m_l}
\end{align}
for all levels $l=0,\ldots,L_{-}$ separately and $x=(x_i)_{i\in E_l}$, $x_i \in \{-1,1\}$. Without loss of generality $\sum_{i\in E_l}x_i=m_l$ (otherwise both sides in the last display are zero). Since $\tau_l$ is a  uniform random permutation on $E_l$, it is clear that $\P{W_{\tau_l}=x}=\binom{n_l}{m_l}^{-1}$. On the other hand,
\begin{align}
  \P{ (X_i)_{i\in E_l}=x \,\mid\, M_l = m_l }
  = \frac{\P{ (X_i)_{i\in E_l}=x , M_l = m_l }}{\P{M_l = m_l}}
  &= \frac{\P{ (X_i)_{i\in E_l}=x }}{\P{M_l = m_l}}
  \\= \frac{p_{l}^{m_{l}}(1 - p_{l})^{n_{l} - m_{l}}}{ \binom{n_l}{m_l} p_{l}^{m_{l}}(1 - p_{l})^{n_{l} - m_{l}}}
  = \binom{n_l}{m_l}^{-1}.
\end{align}
\end{proof}

\begin{proof}[Proof of \Cref{lem:stoch_mono}]
  Let $m=(m_0,\ldots,m_{L_{-}})$, $\widetilde{m}=(\widetilde{m}_0,\ldots,\widetilde{m}_{L_{-}})$ with integers $0\leq m_l \leq \widetilde{m}_{l}\leq n$ for all $l = 0, \ldots, L_{-}$ such that $\sum_{l=0}^{L_{-}} m_l $, $\sum_{l=0}^{L_{-}} \widetilde{m}_l \leq n$. Consider two graphs $G=(V,E,(W_i)_{i \in E})$ and $\widetilde{G}=(V,E,(\widetilde{W_i})_{i\in E})$ with edge weights $W_i, \widetilde{W}_i \in \left\{ -1, 1 \right\}$ such that $\sum_{i \in E_l} B_l = m_l, \sum_{i \in E_l} \widetilde{B}_l = \widetilde{m}_l$ for all $l = 0, \ldots, L_{-}$, where $B_l=\frac{1-W_i}{2}$, $\widetilde{B}_l = \frac{1-\widetilde{W}_i}{2}$. Define further
  \begin{align}
    R_r &= \sum_{(i,j,k) \in \triangle} \I{ B_{\pi(i)} + B_{\pi(j)} + B_{\pi(k)} \geq r}, \\
    \widetilde{R}_r &= \sum_{(i,j,k) \in \triangle} \I{ \widetilde{B}_{\pi(i)} + \widetilde{B}_{\pi(j)} + \widetilde{B}_{\pi(k)} \geq r},
  \end{align}
  for $r = 1, 2, 3$ and with the random permutation $\pi$. $R_r$ and $\widetilde{R_r}$ count the number of triangles with at least $r$ negative edges in $G$ and $\widetilde{G}$ (up to overcounting factors). According to \Cref{lem:cond_dist}, we have
  \begin{align}
    \mathcal{L}( R_3, R_2, R_1) &= \mathcal{L}( H T \,\mid\, M = m), \\
    \mathcal{L}( \widetilde{R}_3, \widetilde{R}_2, \widetilde{R}_1) &= \mathcal{L}( H T \,\mid\, M = \widetilde{m}).
  \end{align}
  Since $R_r, \widetilde{R}_{r}$ are defined on the same probability space, it is sufficient to show $R_r \leq \widetilde{R}_r$ to prove stochastic monotonicity (cf. Remark 2.1 of \citep{Janson:2007dx}). Moreover, it is enough to consider the special case that $\widetilde{m}_1 = m_1 + 1$ and $m_l = \widetilde{m}_l$ for $l = 2, \ldots, L_{-}$. Since $\pi$ restricted to $E_l$ is a uniform random permutation, we can further assume without loss of generality that $W_{i_0} = -1, \widetilde{W}_{i_0} = 1$ for some $i_0 \in E_1$ and $W_i = \widetilde{W}_i$ otherwise. We claim that for $(i,j,k) \in \triangle$
  \begin{align}
    \I{B_{\pi(i)} + B_{\pi(j)} + B_{\pi(k)}  \geq r} \leq \I{\widetilde{B}_{\pi(i)} + \widetilde{B}_{\pi(j)} + \widetilde{B}_{\pi(k)} \geq r}.
  \end{align}
  If $\pi(i),\pi(j),\pi(k)$ are all different from $i_0$ then equality holds trivially. Otherwise, assume that $\pi(i)= i_0$ and $\pi(j),\pi(k) \neq i_0$. Then, the inequality reduces to
  \begin{align}
  \I{B_{\pi(j)} + B_{\pi(k)} \geq r} \leq \I{\widetilde{B}_{\pi(j)} + \widetilde{B}_{\pi(k)} \geq r - 1},
  \end{align}
  which is clearly true. Thus, $R_r \leq \widetilde{R}_r$ for $r = 1, 2,3$.
\end{proof}

\begin{proof}[Proof of \Cref{lem:old_test_fails}]
  We already have $\P{U^{(n)}_{\pi}\leq c_{\alpha,\pi}^{(n)}} = \alpha$ by definition of the critical value.
  Thus, we only have to show $\P{U^{(n)}_{\pi}\leq c_{\alpha,\tau}^{(n)}} \rightarrow 1$. Let $\mu^{(n)}_{\tau}=\E{U^{(n)}_{\tau}}$, $(\sigma^{(n)}_{\tau})^2=\V{U^{(n)}_{\tau}}$ and define similarly $\mu^{(n)}_{\pi}$, $\sigma^{(n)}_{\pi}$ with respect to $U^{(n)}_{\pi}$. 
  From \Cref{cor:an_uniform}, we know that $U_{\tau}^{(n)}$ has a limiting normal distribution. Thus, 
  for large $n$ and $\alpha\leq 1/2$ we can therefore approximate the $\alpha$-critical value of the old test, $c_{\alpha,\tau}^{(n)}$, by $\mu^{(n)}_{\tau} + z_{\alpha} \sigma^{(n)}_{\tau}$ with $0\geq z_{\alpha} \coloneqq \Phi^{-1}(\alpha)$. Moreover, according to \Cref{thm:main}, the Type-I error is given by
  \begin{align}
    \P{ U^{(n)}_{\pi} \leq c_{\alpha,\tau}^{(n)} }
    \asymp \Phi\left(\frac{c_{\alpha,\tau}^{(n)} - \mu^{(n)}_{\pi}}{\sigma^{(n)}_{\pi}}\right)
    \asymp \Phi\left(\frac{\mu^{(n)}_{\tau} + z_{\alpha} \sigma^{(n)}_{\tau} - \mu^{(n)}_{\pi}}{\sigma^{(n)}_{\pi}}\right),
  \end{align}
  where $a_n \asymp b_n$ means $\lim_{n \to \infty} \frac{a_n}{b_n} = 1$.
  Since $n^{-1/2}\sigma_{\pi}^{(n)}\rightarrow \sigma_{\pi}$ and $n^{-1/2}\sigma_{\tau}^{(n)}\rightarrow \sigma_{\tau}$ for some $\sigma_{\pi},\sigma_{\tau}>0$, we have
  \begin{align}
  	\frac{\mu^{(n)}_{\tau} + z_{\alpha} \sigma^{(n)}_{\tau} - \mu^{(n)}_{\pi}}{\sigma^{(n)}_{\pi}} = 		\frac{\mu^{(n)}_{\tau} - \mu^{(n)}_{\pi}}{n^{1/2}\sigma_{\pi}(1+o(1))} + z_{\alpha} \frac{\sigma_{\tau}}{\sigma_{\pi}}\left(1+o(1)\right).
  \end{align}
  Again, by \Cref{thm:main} and  \Cref{cor:an_uniform}, approximating the means of $U^{(n)}_{\pi}$ and $U^{(n)}_{\tau}$ by the corresponding means in the Rademacher model,  we have for large $n$,
  \begin{align}
      & \mu^{(n)}_{\tau} - \mu^{(n)}_{\pi} \approx 
      \frac{1}{2} \sum_{(i,j,k) \in E^{3}} \triangle_{ijk} \left(p(1-p)^2-p_{\varepsilon_i} (1-p_{\varepsilon_j}) (1-p_{\varepsilon_k})\right) \\  
      & \qquad + \frac{1}{6} \sum_{(i,j,k) \in E^{3}} \triangle_{ijk} \left(p^3-p_{\varepsilon_i} p_{\varepsilon_j} p_{\varepsilon_k}\right).
  \end{align}
  Observe that
  \begin{align}
  	\frac{1}{2}p_{\varepsilon_i}(1-p_{\varepsilon_j})(1-p_{\varepsilon_k}) + \frac{1}{6}p_{\varepsilon_i}p_{\varepsilon_j}p_{\varepsilon_k} = \frac{1}{2}p_{\varepsilon_i}-\frac{1}{2}p_{\varepsilon_i}p_{\varepsilon_j}-\frac{1}{2}p_{\varepsilon_i}p_{\varepsilon_k} + \frac{2}{3}p_{\varepsilon_i}p_{\varepsilon_j}p_{\varepsilon_k},
  \end{align}
  which can be bounded by $\frac{1}{2}p_{\varepsilon_i}$. Hence, for some constant $c>0$
  \begin{align}
      & \mu^{(n)}_{\tau} - \mu^{(n)}_{\pi} \geq \frac{c}{2} \sum_{(i,j,k) \in E^{3}} \triangle_{ijk} \left(p(1-p)^2-p_{\varepsilon_i}\right) \\  
      & = \frac{c}{2} \left(p(1-p)^2 \sum_{i \in E} \varepsilon_i - \sum_{i \in E} p_{\varepsilon_i}\varepsilon_i \right) = \frac{c}{2} \left(p(1-p)^2 \sum_{i \in E} \varepsilon_i - \sum_{i \in E^{-}} \varepsilon_i \right)\\
      & \geq \frac{mc}{2} \left((1-p)^2 \frac{\sum_{i \in E} \varepsilon_i}{n} - \frac{ \sum_{i \in E^{-}} \varepsilon_i }{m} \right)
      \geq cm^{1/2}\log{m},
  \end{align}
  for a different $c>0$ according to \cref{eqn:difference_means}, because $p_{\varepsilon_i}=m_{\varepsilon_i}/n_{\varepsilon_i}$ and $p=m/n$. Consequently, ${m^{-1/2}} \left(\mu^{(n)}_{\tau} - \mu^{(n)}_{\pi}\right)$ converges to $\infty$ and so $\P{ U^{(n)}_{\pi} \leq c_{\alpha,\tau}^{(n)}}\rightarrow 1$. 
\end{proof}